\documentclass[11pt]{article}

\usepackage{amssymb,amsmath,amsfonts}
\usepackage{graphicx,xcolor,enumitem}
\usepackage{epsfig}
\usepackage{amsthm}
\usepackage{bm}
\usepackage{subcaption}

\usepackage[a4paper, twoside, hmarginratio=1:1, vmarginratio=1:1, left=1in,top=1in]{geometry}

\RequirePackage[breaklinks=true, hidelinks]{hyperref}
\usepackage{breakcites}

\newcommand{\p}{\mathbb{P}}
\newcommand{\F}{\mathbb{F}}
\newcommand{\E}{\mathbb{E}}
\newcommand{\tildeE}{\mathbb{\tilde E}}
\newcommand{\R}{\mathbb{R}}
\newcommand{\cF}{{\mathcal F}}
\newcommand{\cU}{{\mathcal U}}

\newcommand{\id}{{\rm id}}
\newcommand{\as}{\mbox{{\rm a.s.}}}

\newtheorem{theorem}{Theorem}

\newtheorem{assumption}[theorem]{Assumption}
\newtheorem{corollary}[theorem]{Corollary}

\newtheorem{definition}[theorem]{Definition}

\newtheorem{lemma}[theorem]{Lemma}

\newtheorem{remark}[theorem]{Remark}

\theoremstyle{definition}

\numberwithin{equation}{section}
\numberwithin{theorem}{section}

\begin{document}

\title{Mean-variance portfolio selection under Volterra Heston model}
\author{Bingyan Han \thanks{Department of Statistics, The Chinese University of Hong Kong, Hong Kong, byhan@link.cuhk.edu.hk}
\and Hoi Ying Wong\thanks{Department of Statistics, The Chinese University of Hong Kong, Hong Kong, hywong@cuhk.edu.hk}
}
\date{January 24, 2020}

\maketitle

\begin{abstract}
Motivated by empirical evidence for rough volatility models, this paper investigates continuous-time mean-variance (MV) portfolio selection under the Volterra Heston model. Due to the non-Markovian and non-semimartingale nature of the model, classic stochastic optimal control frameworks are not directly applicable to the associated optimization problem. By constructing an auxiliary stochastic process, we obtain the optimal investment strategy, which depends on the solution to a Riccati-Volterra equation. The MV efficient frontier is shown to maintain a quadratic curve. Numerical studies show that both roughness and volatility of volatility materially affect the optimal strategy.
\\[2ex] 
\\[2ex]
\end{abstract}

\section{Introduction}
There has been a growing interest in studying rough volatility models \cite{gatheral2018volrough,eleuch2019char,guennoun2018asym}. Rough volatility models are stochastic volatility models whose trajectories are rougher than the paths of a standard Brownian motion in terms of the H\"older regularity. Specifically, when the H\"older regularity is less than 1/2, the stochastic path is regarded as rough. The roughness is closely related to the Hurst parameter $H$. This paper focuses on the Volterra Heston model, whose probabilistic characterization does not involve the rough paths theory \cite{eleuch2019char}. 

Rough volatility models are attractive because they capture the dynamics of historical and implied volatilities remarkably well with only a few additional parameters. Investigations of the time series of the realized volatility\footnote{See, for example, Oxford-Man Institute's realized library at \url{https://realized.oxford-man.ox.ac.uk/data}}  from high frequency data estimate the Hurst parameter $H$ to be near $0.1$, which is much smaller than the $0.5$ for the standard Brownian motion. The Hurst parameter is used to reflect the memoryness of a time series and is associated with the roughness of the fractional Brownian motion (fBM). The smaller the $H$, the rougher the time series model. Therefore, the empirical finding suggests a rougher realized path of volatility than the standard Brownian motion. Although previous studies have found a long memory property within realized volatility series, it is shown in \cite{gatheral2018volrough} that rough volatility models can generate the illusion of a long memory. However, the simulated paths with a small Hurst parameter resemble the realized ones. 

Rough volatility models also better capture the term structure of an implied volatility surface, especially for the explosion of at-the-money (ATM) skew when maturity goes to zero. More precisely, let $\sigma_{BS}(k, \tau)$ be the implied volatility of an option where $k$ is the log-moneyness and $\tau$ is the time to expiration. The ATM skew at maturity $\tau$ is defined by
\begin{equation}
\phi(\tau) \triangleq \Big| \frac{\partial \sigma_{BS}(k, \tau)}{\partial k} \Big|_{k=0}.
\end{equation}
Empirical evidence shows that the ATM skew explodes when $\tau \downarrow 0$. However, conventional volatility models such as the Heston model \cite{heston1993closed} generate a constant ATM skew for a small $\tau$. If the volatility is modeled by a fractional Brownian motion, then the ATM skew has an asymptotic property \cite{fukasawa2011asymptotic},
\begin{equation}
\phi(\tau) \approx \tau^{H - 1/2}, \text{ when } \tau \downarrow 0,
\end{equation}
where $H$ is the Hurst parameter. Rough volatility models can fit the explosion remarkably well by simply adjusting the $H$. 

Recent advances offer elegant theoretical foundations for rough volatility models. We note the martingale expansion formula for implied volatility \cite{fukasawa2011asymptotic}, asymptotic analysis of fBM \cite[Section 3.3]{fukasawa2011asymptotic}, the microstructural foundation of rough Heston models by scaling the limit of proper Hawkes processes \cite{eleuch2018micro}, the closed-form characteristic function of rough Heston models up to the solution of a fractional Riccati equation \cite{eleuch2019char}, and the hedging strategy for options under rough Heston models \cite{eleuch2018perfect}. In this paper, we are particularly interested in the affine Volterra processes \cite{abi2017affine} because these models embrace rough Heston model \cite{eleuch2019char} as a special case. The characteristic function in \cite{eleuch2019char} is extended to the exponential-affine transform formula in terms of Riccati-Volterra equations \cite{abi2017affine}.  Affine Volterra processes are applied to finance problems in \cite{keller2018affine}. In addition, an alternative rough version of the Heston model is introduced in \cite{guennoun2018asym}, where some asymptotic results are derived.

While the rough volatility literature focuses on option pricing, only a few works contribute to portfolio optimization such as \cite{fouque2018fast,fouque2018fractional,bauerle2018protfolio}. All of them consider utility maximization. To the best of our knowledge, this is the first paper to consider the mean-variance (MV) portfolio selection under a rough stochastic environment. The MV criterion in portfolio selection pioneered by Markowitz's seminal work is the cornerstone of the modern portfolio theory. We cannot give a full list of research outputs related to this Nobel Prize winning work, but mention contributions in continuous-time settings \cite{zhou2000lq,lim2002complete,lim2004incomplete,cerny2008mean,jeanblanc2012mean,shen2015mean} as important references.

\subsection{Major contributions}

We formulate the MV portfolio selection under the Volterra Heston models in a reasonably rigorous manner. As pointed out by \cite{abi2017affine,keller2018affine}, the Volterra Heston model (\ref{vol})-(\ref{stock}) has a unique in law weak solution, but its pathwise uniqueness is still an open question in general. This enforces us to consider the MV problem under a general filtration $\F$ that satisfies the usual conditions but may not be the augmented filtration generated by the Brownian motion. A similar general setting also appears in \cite{jeanblanc2012mean}. We emphasize that the probability basis and Brownian motions are always fixed for the problem in Section \ref{Sec:MV}. Therefore, our formulation is still considered to be a {\it strong formulation}, because the filtered probability space and Brownian motions are not parts of the control.

Under such a problem formulation, we construct in Section \ref{Sec:Sol} an auxiliary stochastic process $M_t$  to solve the MV portfolio selection by completion of squares. Several properties of $M_t$ are derived in Theorem \ref{Thm:M}, which is a main result of this paper. Like \cite{eleuch2019char,eleuch2018perfect,abi2017affine}, we encounter difficulties due to the non-Markovian and non-semimartingale structure of the Volterra Heston model (\ref{vol})-(\ref{stock}). Inspired by the exponential-affine formulas in \cite{abi2017affine,eleuch2019char}, the process $M_t$ is constructed upon the forward variance under a proper alternative measure. The explicit solution for the optimal investment strategy is obtained in Theorem \ref{Thm:Sol}.

Under the rough Heston model, we investigate the impact of roughness on the optimal investment strategy $u^*$. Recently, a trading strategy has been proposed to leverage the information of roughness \cite{glasserman2019buy}. The strategy longs the roughest stocks and shorts the smoothest stocks. Excess returns from this strategy are not fully explained by standard factor models like the CAPM model and Fama-French model. We examine this trading signal under the MV setting. Our theory predicts that the effect of roughness on investment strategy is opposite under different volatility of volatility (vol-of-vol). We also discuss the roughness effect on the efficient frontier.

The rest of the paper is organized as follows. Section \ref{Sec:Model} presents the Volterra Heston model and some useful properties. We discuss a related Riccati-Volterra equation. We then formulate the MV portfolio selection problem in Section \ref{Sec:MV} and solve it explicitly in Section \ref{Sec:Sol}. Numerical illustrations are given in Section \ref{Sec:Numerical}. Section \ref{Sec:Conclusion} concludes the paper. The existence and uniqueness of the solution to Riccati-Volterra equations are summarized in Appendix \ref{Appendix}. An auxiliary result used in Theorem \ref{Thm:M} is proved in Appendix \ref{App:Pos}.

\section{The Volterra Heston model}\label{Sec:Model}
Our problem is defined under a given complete probability space $(\Omega, \cF, \p)$, with a filtration $\F = \{ \cF_t \}_{ 0 \leq t \leq T}$ satisfying the usual conditions, supporting a two-dimensional Brownian motion $W = (W_1, W_2)$. The filtration $\F$ is not necessarily the augmented filtration generated by $W$; thus, it can be a strictly larger filtration. This consideration is different from some previous studies like \cite{lim2002complete,lim2004incomplete,shen2015mean} but is consistent with \cite{jeanblanc2012mean} for the MV hedging problem under a general filtration. This consideration is important because the stochastic Volterra equation (\ref{vol})-(\ref{stock}) only has a unique in law weak solution but its strong uniqueness is still an open question in general. Recall that for stochastic differential equations, $X$ is referred to as a strong solution if it is adapted to the augmented filtration generated by $W$, and a weak solution otherwise. For a weak solution, the driving Brownian motion $W$ is also a part of the solution \cite[Chapter IX]{ry1999book}. Therefore, $\F$ cannot be simply chosen as the augmented filtration generated by $W$, as extra information may be needed to construct a solution to (\ref{vol})-(\ref{stock}).

To proceed, we introduce a kernel $K(\cdot) \in L^2_{loc} (\R_+, \R)$, where $\R_+ = \{ x \in \R | x \geq 0\}$,  and make the following standing assumption throughout the paper, in line with \cite{abi2017affine,keller2018affine}. A function $f$ is called completely monotone on $(0, \infty)$, if it is infinitely differentiable on $(0, \infty)$ and $(-1)^k f^{(k)}(t) \geq 0$ for all $ t > 0 $, and $k = 0, 1, ...$.
\begin{assumption}\label{Assum:K}
	$K$ is strictly positive and completely monotone on $(0, \infty)$. There is $\gamma \in(0,2]$, such that $\int_{0}^{h} K(t)^{2} d t=O\left(h^{\gamma}\right)$ and $\int_{0}^{T}(K(t+h)-K(t))^{2} d t=O\left(h^{\gamma}\right)$ for every $T<\infty$. 
\end{assumption}

The convolutions $K*L$ and $L*K$ for a measurable kernel $K$ on $\R_+$ and a measure $L$ on $\R_+$ of locally bounded variation are defined by
\begin{equation}
(K*L)(t) = \int_{[0,t]} K(t-s)L(ds) \quad \text{and} \quad (L*K)(t) = \int_{[0,t]} L(ds)K(t-s)
\end{equation}
for $t>0$ under proper conditions. The integral is extended to $t=0$ by right-continuity if possible. If $F$ is a function on $\R_+$, let
\begin{equation}
(K*F)(t) = \int_0^t K(t-s) F(s) ds.
\end{equation}

Let $W$ be a $1$-dimensional continuous local martingale. The convolution between $K$ and $W$ is defined as
\begin{equation}
(K*dW)_t = \int_0^t K(t-s)dW_s.
\end{equation}

A measure $L$ on $\R_+$ is called {\em resolvent of the first kind} to $K$, if
\begin{equation}
K*L = L*K \equiv \id.
\end{equation}
The existence of a resolvent of the first kind is shown in \cite[Theorem 5.5.4]{gripenberg1990volterra} under the complete monotonicity assumption, imposed in Assumption \ref{Assum:K}. Alternative conditions for the existence are given in \cite[Theorem 5.5.5]{gripenberg1990volterra}.

A kernel $R$ is called the {\em resolvent} or {\em resolvent of the second kind} to $K$ if
\begin{equation}
K*R = R*K = K - R.
\end{equation}
The resolvent always exists and is unique by \cite[Theorem 2.3.1]{gripenberg1990volterra}.

Further properties of these definitions can be found in \cite{gripenberg1990volterra,abi2017affine}. Although the same notion can be defined for higher dimensions and in matrix form, it suffices for us to consider the scalar case. Commonly used kernels \cite{abi2017affine} summarized in Table \ref{Tab:Kernel} satisfy Assumption \ref{Assum:K} once $c>0$, $\alpha \in (1/2, 1]$, and $\beta \geq 0$.

\begin{table}[h!]
	\centering
	\begin{tabular}{c c c c }
		\hline
		& $K(t)$ & $R(t)$ & $L(dt)$ \\ 
		\hline \\[0.5ex]
		Constant		& $c$ & $ce^{-ct}$ & $c^{-1} \delta_0(dt)$\\ \\
		Fractional (Power-law)		& $c\,\frac{t^{\alpha-1}}{\Gamma(\alpha)}$ & $ct^{\alpha-1}  E_{\alpha, \alpha} (-ct^{\alpha})$ & $c^{-1}\,\frac{t^{-\alpha}}{\Gamma(1-\alpha)}dt$\\ \\
		Exponential	& $ce^{-\beta t}$ & $ce^{-\beta t}e^{-ct}$ & $c^{-1}(\delta_0(dt)  + \beta\,dt)$ \\\\
		\hline
	\end{tabular}
	\caption{Examples of kernels $K$ and their resolvents $R$ and $L$ of the second and first kind. $E_{\alpha,\beta}(z)=\sum_{n=0}^\infty \frac{z^n}{\Gamma(\alpha n+\beta)}$ is the Mittag--Leffler function. See \cite[Appendix A1]{eleuch2019char} for its properties. The constant $c \neq 0$.}
	\label{Tab:Kernel}
\end{table}

The variance process within the Volterra Heston model is defined as
\begin{equation}\label{vol}
V_{t}=V_{0}+ \kappa \int_{0}^{t} K(t-s)\left(\phi -V_{s}\right) d s + \int_{0}^{t} K(t-s) \sigma \sqrt{V_{s}} d B_{s},
\end{equation}
where $ dB_s = \rho dW_{1s} + \sqrt{1 - \rho^2} dW_{2s} $ and $V_0, \kappa, \phi$, and $\sigma$ are positive constants. The correlation $\rho$ between stock price and variance is also constant. As documented in \cite{gatheral2018volrough}, the general overall shape of the implied volatility surface does not change significantly, indicating that it is still acceptable to consider a variance process whose parameters are independent of stock price and time. The rough Heston model in \cite{eleuch2019char,eleuch2018perfect} becomes a special case of (\ref{vol}) once $K(t) = \frac{t^{\alpha-1}}{\Gamma(\alpha)}$. Another rough version of the Heston model studied in \cite{guennoun2018asym} is adopted to investigate the power utility maximization \cite{bauerle2018protfolio}.

Following  \cite{abi2017affine} and \cite{kraft2005optimal,cerny2008mean,zeng2013portfolio,shen2015square}, the risky asset (stock) price $S_t$ is assumed to follow
\begin{equation}\label{stock}
dS_t = S_t (r_t + \theta V_t) dt + S_t \sqrt{V_t} dW_{1t}, \quad S_0 > 0,
\end{equation}
with a deterministic bounded risk-free rate $r_t>0$ and constant $\theta \neq 0$. The market price of risk, or risk premium, is then given by $\theta \sqrt{V_t}$. The risk-free rate $r_t >0$ is the rate of return of a risk-free asset available in the market.

We take the existence and uniqueness result from \cite[Theorem 7.1]{abi2017affine} and restate it as follows.
\begin{theorem}\label{Thm:SVSol}
	(\cite[Theorem 7.1]{abi2017affine}) Under Assumption \ref{Assum:K}, the stochastic Volterra equation (\ref{vol})-(\ref{stock}) has a unique in law $\R_+ \times \R_+$-valued continuous weak solution for any initial condition $(S_0, V_0) \in \R_+ \times \R_+$.
\end{theorem}
\begin{remark}
	Our model (\ref{vol})-(\ref{stock}) is defined under the physical measure, whereas the option pricing model of  \cite[Equations (7.1)-(7.2)]{abi2017affine} is under a risk-neutral measure with a zero risk-free rate. However, the proofs are almost identical because the affine structure is maintained and $S$ is determined by $V$.
\end{remark}
\begin{remark}
	For strong uniqueness, we mention \cite[Proposition B.3]{abi2019multifactor} as a related result with kernel $K \in C^1([0, T], \R)$ and \cite[Proposition 8.1]{mytnik2015uniqueness} for certain Volterra integral equations with smooth kernels. However, the strong uniqueness of (\ref{vol})-(\ref{stock}) is left open for singular kernels. For weak solutions, it is free to construct the Brownian motion as needed. However, the MV objective only depends on the mathematical expectation for the distribution of the processes. In the sequel, we will only work with a version of the solution to (\ref{vol})-(\ref{stock}) and fix the solution $(S, V, W_1, W_2)$, as other solutions have the same law.
\end{remark}

The following condition enables us to verify the admissibility of the optimal strategy. To be more precise about the constant $a$, \eqref{Eq:Const_a} gives an explicit sufficient large value needed.
\begin{assumption}\label{Assum:V}
	$\E \Big[ \exp \big( a \int^T_0 V_s ds \big) \Big] < \infty$ for a large enough constant $a > 0$.
\end{assumption}

To verify that Assumption \ref{Assum:V} holds under reasonable conditions, we consider the Riccati-Volterra equation (\ref{Eq:g}) for $g(a, t)$ as follows:
\begin{equation}\label{Eq:g}
g(a, t) = \int^t_0 K(t-s) \big[ a - \kappa g(a,s) + \frac{\sigma^2}{2} g^2(a, s) \big] ds.
\end{equation}
The existence and uniqueness of the solution to  (\ref{Eq:g}) are given in Lemmas \ref{Lem:g} and \ref{Lem:gfractional}.

\begin{theorem}\label{Thm:ExpV}
	Suppose Assumption \ref{Assum:K} holds and the Riccati-Volterra equation (\ref{Eq:g}) has a unique continuous solution on $[0, T]$, then
	\begin{equation}
	\E \Big[ \exp \big( a \int^T_0 V_s ds \big) \Big] =  \exp\Big[ \kappa \phi \int^T_0 g(a, s) ds + V_0 \int^T_0 \big[ a - \kappa g(a, s) + \frac{\sigma^2}{2} g^2(a, s) \big] ds \Big] < \infty.
	\end{equation}
	Moreover, denote $L$ as the resolvent of the first kind to $K$, then
	\begin{equation}
	\E \Big[ \exp \big( a \int^T_0 V_s ds \big) \Big] = \exp \Big[ \kappa \phi \int^T_0 g(a, s) ds + V_0 \int^T_0 g(a, T-s) L(ds) \Big].
	\end{equation}
\end{theorem}
\begin{proof}
	Note $g(a, t)$ in \eqref{Eq:g} corresponds to \cite[Equation (4.3)]{abi2017affine} with $u = 0$ and $f = a$. \cite[Theorem 4.3]{abi2017affine} shows the equivalence between \cite[Equation (4.4)]{abi2017affine} and  \cite[Equation (4.6)]{abi2017affine}. For $t = T$, the expressions in \cite[Equation (4.4)-(4.6)]{abi2017affine} indicate that
	\begin{equation}
	a \int^T_0 V_s ds  = Y_0 - \frac{\sigma^2}{2} \int^T_0 g^2(a, T-s) V_s ds + \sigma \int^T_0  g(a, T-s) \sqrt{V_s} dB_s,
	\end{equation}
	with
	\begin{equation}
	Y_0 = \kappa \phi \int^T_0 g(a, s) ds + V_0 \int^T_0 \big[ a - \kappa g(a, s) + \frac{\sigma^2}{2} g^2(a, s) \big] ds.
	\end{equation}
	
	As $g(a, \cdot)$ is continuous on $[0, T]$ and therefore bounded, $\exp\big( - \frac{\sigma^2}{2} \int^t_0 g^2(a, T-s) V_s ds + \sigma \int^t_0  g(a, T-s) \sqrt{V_s} dB_s\big)$ is a martingale by \cite[Lemma 7.3]{abi2017affine}. Therefore,
	\begin{equation}
	\E \Big[ \exp \big( a \int^T_0 V_s ds \big) \Big] = \exp(Y_0) = \exp\Big[ \kappa \phi \int^T_0 g(a, s) ds + V_0 \int^T_0 \big[ a - \kappa g(a, s) + \frac{\sigma^2}{2} g^2(a, s) \big] ds \Big].
	\end{equation}
	Note that $K * L = \id $ implies
	\begin{equation}
	\int^T_0 \big[ a - \kappa g(a, s) + \frac{\sigma^2}{2} g^2(a, s) \big] ds  =  \int^T_0 g(a, T-s) L(ds).
	\end{equation}
	The result follows.
\end{proof}
Theorem \ref{Thm:ExpV} recovers the same expression for $\E \Big[ \exp \big( a \int^T_0 V_s ds \big) \Big]$ in \cite[Theorem 3.2]{eleuch2018perfect}. We stress that the proof circumvents the use of the Hawkes processes. In addition, we mention \cite{gerhold2018moment}, which examines the moment explosions in the rough Heston model, as a related reference.

\section{Mean-variance portfolio selection}\label{Sec:MV}

Let $u_t \triangleq \sqrt{V_t} \pi_t$ be the investment strategy, where $\pi_t$ is the amount of wealth invested in the stock. Then wealth process $X_t$ satisfies 
\begin{equation}\label{Eq:wealth}
d X_t = \big(r_t X_t + \theta \sqrt{V_t} u_t \big) dt + u_t dW_{1t}, \quad X_0 = x_0 > 0.
\end{equation}

\begin{definition}
	An investment strategy $u(\cdot)$ is said to be admissible if 
	\begin{enumerate}[label={(\arabic*).}]
		\item $u(\cdot)$ is $\F$-adapted;
		\item $\E\Big[ \Big(\int^T_0 |\sqrt{V_t} u_t |dt \Big)^2 \Big] < \infty$ and $\E\Big[ \int^T_0 |u_t|^2 dt \Big] < \infty$; and
		\item the wealth process (\ref{Eq:wealth}) has a unique solution in the sense of \cite[Chapter 1, Definition 6.15]{yong1999book}, with $\p$-$\as$ continuous paths.
	\end{enumerate}
	The set of all of the admissible investment strategies is denoted as $\cU$.
\end{definition}

\begin{remark}
	In Condition (1), $\F$ is possibly strictly larger than the Brownian filtration of $W = (W_1, W_2)$, which means that extra information in addition to $W$ can be used to construct an admissible strategy. In general, $u$ can rely on a local $\p$-martingale that is strongly $\p$-orthogonal to $W$. See hedging strategy (3.6) in \cite[Theorem 3.1]{jeanblanc2012mean} for such examples. However, our optimal strategy $u^*$ turns out to only depend on the variance $V$ and Brownian motion $W$, as shown in Theorem \ref{Thm:Sol}.
\end{remark}

\begin{remark} We emphasize once again that the underlying probability space and Brownian motions are not parts of our control. Therefore, our formulation should still be referred to as a strong formulation. Readers may refer to \cite[Chapter 2, Section 4]{yong1999book} for discussions of the difference between strong and weak formulations of stochastic control problems.
\end{remark}

The MV portfolio selection in continuous-time is the following problem\footnote{There are several equivalent formulations.}.
\begin{equation}\label{Eq:obj}
\left\{\begin{array}{l}{ 
	\min _{ u(\cdot) \in \cU} J \left(x_{0} ; u(\cdot)\right) = \E \left[ (X_T - c)^2 \right]}, \\
\text{ subject to } \E[X_T] = c, \\
(X(\cdot), u(\cdot)) \text { satisfy (\ref{Eq:wealth})}.
\end{array}\right.
\end{equation}
The constant $c$ is the target wealth level at the terminal time $T$. We assume $c\geq x_0 e^{\int^T_0 r_s ds}$ following \cite{lim2002complete,lim2004incomplete,shen2015mean}. Otherwise, a trivial strategy that puts all of the wealth into the risk-free asset can dominate any other admissible strategy. The MV problem is said to be feasible for $c \geq x_0 e^{\int^T_0 r_s ds}$ if there exists a $u(\cdot) \in \cU$ that satisfies $\E[X_T] = c$. Note that $r_t > 0$ is deterministic and  $\E[\int^T_0 V_t dt] > 0$. It is then clear that the feasibility of our problem is guaranteed for any $c\geq x_0 e^{\int^T_0 r_s ds}$ by a slight modification to the proof in \cite[Propsition 6.1]{lim2004incomplete}.

As Problem (\ref{Eq:obj}) has a constraint, it is equivalent to the following max-min problem \cite{luenberger1968opt}.
\begin{equation}\label{Eq:maxminobj}
\left\{\begin{array}{l}{ 
	\max _{\eta \in \R } \min _{ u(\cdot) \in \cU} J\left(x_{0} ; u(\cdot)\right) = \E \left[(X_T - (c-\eta))^{2}\right] - \eta^{2}}, \\ 
(X(\cdot), u(\cdot)) \text { satisfy (\ref{Eq:wealth})}.
\end{array}\right.
\end{equation}
Let $\zeta = c - \eta$ and consider the inner Problem (\ref{Eq:innerobj}) of (\ref{Eq:maxminobj}) first.
\begin{equation}\label{Eq:innerobj}
\left\{\begin{array}{l}{ 
	\min _{ u(\cdot) \in \cU} J\left(x_{0} ; u(\cdot)\right) = \E \left[(X_T - \zeta )^{2}\right] - \eta^{2} }, \\ 
(X(\cdot), u(\cdot)) \text { satisfy (\ref{Eq:wealth})}.
\end{array}\right.
\end{equation}

\section{Optimal investment strategy}\label{Sec:Sol}
To solve Problem (\ref{Eq:innerobj}), we introduce a new probability measure $ \tilde \p$ by
\begin{equation}\label{Eq:tildeP}
\left. \frac{d \tilde \p}{d\p} \right|_{\cF_t} = \exp\Big( - 2 \theta^2 \int^t_0 V_s ds - 2 \theta \int^t_0 \sqrt{V_s} dW_{1s} \Big),
\end{equation}
where the stochastic exponential is a true martingale \cite[Lemma 7.3]{abi2017affine}. Then $\tilde W_{1t} \triangleq W_{1t}  + 2 \theta \int^t_0 \sqrt{V_s} ds$ is a new Brownian motion under $\tilde \p$. Hence,
\begin{equation}
V_{t}=V_{0}+ \int_{0}^{t} K(t-s)\left( \kappa \phi - \lambda V_{s}\right) d s + \int_{0}^{t} K(t-s) \sigma \sqrt{V_{s}} d \tilde B_{s},
\end{equation}
where $\lambda = \kappa + 2\theta \rho \sigma$ and $d \tilde B_s = \rho d\tilde W_{1s} + \sqrt{1- \rho^2} d W_{2s}$.

Denote $\tilde \E[\cdot]$ and $\tilde \E[ \cdot |\cF_t]$ as the $\tilde \p$-expectation and conditional $\tilde \p$-expectation,  respectively. The forward variance under $\tilde \p$ is the conditional $\tilde \p$-expected variance: $ \tildeE \left[V_{s} | \mathcal{F}_{t}\right] \triangleq \xi_{t}(s)$.  The following identity is proven in \cite[Propsition 3.2]{keller2018affine} by an application of \cite[Lemma 4.2]{abi2017affine}. 
\begin{equation}\label{Eq:xi}
\xi_{t}(s) = \tildeE \left[V_{s} | \mathcal{F}_{t}\right]=\xi_{0}(s)+\int_0^t \frac{1}{\lambda} R_{\lambda}(s-u) \sigma \sqrt{V_{u}} d \tilde B_{u},
\end{equation}
where
\begin{equation}
\xi_{0}(s) = \left(1-\int_{0}^{s} R_{\lambda}(u) d u\right) V_{0} + \frac{\kappa \phi}{\lambda} \int_{0}^{s} R_{\lambda}(u) du,
\end{equation}
and $R_\lambda$ is the resolvent of $\lambda K$ such that 
\begin{equation}\label{Eq:R_lambda}
\lambda K * R_\lambda = R_\lambda * ( \lambda K) = \lambda K - R_\lambda.
\end{equation}
If $\lambda = 0$, interpret $R_\lambda/\lambda = K$ and $R_\lambda = 0$.

Consider the stochastic process, 
\begin{equation}\label{Eq:M}
M_t = 2 \exp \Big[ \int^T_t \big(2 r_s - \theta^2 \xi_t(s) + \frac{(1-2\rho^2) \sigma^2}{2} \psi^2(T-s) \xi_t(s) \big) ds \Big],
\end{equation}
where
\begin{equation}\label{Eq:psi}
\psi(t) = \int^t_0 K(t - s) \big[ \frac{(1-2\rho^2) \sigma^2}{2} \psi^2(s)  - \lambda \psi(s) - \theta^2 \big] ds. 
\end{equation}
The existence and uniqueness of the solution to (\ref{Eq:psi}) are established in Lemma \ref{Lem:psi}.

The process $M$ is the key to applying the completion of squares technique in Theorem \ref{Thm:Sol}, inspired by \cite{lim2002complete,lim2004incomplete,shen2015mean}. Heuristically speaking, the non-Markovian and non-semimartingale characteristics of the Volterra Heston model are overcome by considering $M$. The construction of $M$ is based on the following observations. To make a completion of squares, we need an auxiliary process $M$ as an additional stochastic factor in a place consistent with previous studies of MV portfolios under semimartingales. The completion of squares procedure for proving Theorem \ref{Thm:Sol} indicates that $M$ should satisfy \eqref{Eq:dM}. We then link $M$ with the conditional expectation in \eqref{Eq:TransM} via a proper transformation. The exponential-affine transform formula in \cite[Equation (4.7)]{abi2017affine} is applied to obtain \eqref{Eq:M}.

\begin{theorem}\label{Thm:M}
	Assume Assumption \ref{Assum:K} holds and (\ref{Eq:psi}) has a unique continuous solution on $[0, T]$, then $M$ satisfies the following properties.
	\begin{enumerate}[label={(\arabic*).}]
		\item $M_t$ is essentially bounded and $0< M_t < 2e^{2\int^T_t r_s ds}$, $\p$-$\as$, $\forall \; t \in [0, T)$. $M_T = 2$.
		\item Apply It\^o's lemma to $M$ on $t$, then
		\begin{equation}\label{Eq:dM}
		dM_t = \big[ -2 r_t + \theta^2 V_t \big] M_t dt + \big[ 2 \theta \sqrt{V_t} U_{1t} + \frac{U^2_{1t}}{M_t} \big] dt + U_{1t} dW_{1t} + U_{2t} dW_{2t},
		\end{equation}
		where
		\begin{align}
		U_{1t} &= \rho \sigma M_t \sqrt{V_t} \psi(T-t), \label{Eq:U1simp}\\
		U_{2t} &= \sqrt{1-\rho^2} \sigma M_t \sqrt{V_t} \psi(T-t). \label{Eq:U2simp}
		\end{align}
		
		\item 
		\begin{equation}
		M_0 = 2 \exp \Big[  \int^T_0 2 r_s ds + \kappa \phi \int^T_0  \psi(s) ds + V_0 \int^T_0 \big[ \frac{(1-2\rho^2)\sigma^2}{2} \psi^2(s) - \lambda \psi(s) - \theta^2 \big] ds \Big].
		\end{equation}
		Furthermore, for fractional kernel $K(t) = \frac{t^{\alpha-1}}{\Gamma(\alpha)}$,  denote the fractional integral as $I^\alpha \psi(t) = K * \psi(t)$. Then 
		\begin{equation}\label{Eq:fracM0}
		M_0 = 2 \exp \Big[ \int^T_0 2 r_s ds + \kappa \phi I^1 \psi(T) + V_0 I^{1-\alpha}\psi (T) \Big].
		\end{equation}
		\item $\E\Big[ \big(\int^T_0 U^2_{it} dt \big)^{p/2} \Big] < \infty$ for $p \geq 1$ , $ i = 1, 2$.
	\end{enumerate}
\end{theorem}
\begin{proof}
	{\bf Property (1)}.
	
	It is straightforward to see that $M_t > 0$ in (\ref{Eq:M}). As for the upper bound, if $1-2\rho^2 = 0$, note $\int^T_t \xi_t(s) ds > 0$, $\p$-$\as$ by Lemma \ref{Lem:Positive}, then $M_t < 2e^{2\int^T_t r_s ds}$, $\p$-$\as$. If $1-2\rho^2 \neq 0$, we claim 
	\begin{equation}\label{Eq:TransM}
	M^{1-2\rho^2}_t = 2^{1-2\rho^2} \exp\big[2 (1-2\rho^2) \int^T_t r_s ds\big] \tildeE \Big[ \exp \big(- \theta^2(1-2\rho^2) \int^T_t V_s ds\big) \Big| \cF_t \Big].
	\end{equation}
	It is equivalent to show that
	\begin{align}\label{Eq:Mtrans}
	& \tildeE \Big[ \exp \big(- \theta^2(1-2\rho^2) \int^T_t V_s ds\big) \Big| \cF_t \Big] \\
	& =  \exp \Big[ \int^T_t \big( - (1-2\rho^2) \theta^2 \xi_t(s) + \frac{(1-2\rho^2)^2 \sigma^2}{2} \psi^2(T-s) \xi_t(s) \big) ds \Big]. \nonumber
	\end{align}
	Denote $ \tilde \psi = (1 - 2 \rho^2) \psi$. Then $\tilde \psi$ satisfies
	\begin{equation}
	\tilde \psi = K * \big( \frac{\sigma^2}{2} \tilde \psi^2 - \lambda \tilde \psi - (1 - 2\rho^2) \theta^2 \big).
	\end{equation} 
	Therefore, (\ref{Eq:Mtrans}) holds for all $t \in [0, T]$ by \cite[Theorem 4.3]{abi2017affine} applying to $\tilde \psi$. The martingale assumption in \cite[Theorem 4.3]{abi2017affine} is verified by \cite[Lemma 7.3]{abi2017affine}.
	
	If $ 1 - 2 \rho^2 > 0$, then $\tildeE \Big[ \exp \big(- \theta^2(1-2\rho^2) \int^T_t V_s ds\big) \Big| \cF_t \Big] < 1$, $\p$-$\as$, which implies $M_t < 2e^{2\int^T_t r_s ds}$, $\p$-$\as$. $ 1 - 2 \rho^2 < 0$ can be discussed similarly. Property (1) is proved.
	
	\noindent {\bf Property (2).}
	
	Denote $ M_t = 2 e^{Z_t} $ in (\ref{Eq:M}) with proper $Z_t$. We first derive the equation for $dZ_t$. From (\ref{Eq:xi}), apply It\^o's lemma to $\xi_t(s)$ on time $t$ and get
	\begin{equation}
	d \xi_t(s) = \frac{1}{\lambda} R_{\lambda}(s-t) \sigma \sqrt{V_t} d \tilde B_t.
	\end{equation}
	
	Then
	\begin{align*}
	d Z_t  = &\big[ - 2 r_t + \theta^2 V_t - \frac{(1-2\rho^2) \sigma^2}{2} \psi^2(T-t) V_t \big] dt \\
	&  - \theta^2 \int^T_t   \frac{1}{\lambda} R_{\lambda}(s-t) \sigma \sqrt{V_t} d \tilde B_t ds +  \frac{(1-2\rho^2) \sigma^2}{2} \int^T_t  \psi^2(T-s) \frac{1}{\lambda} R_{\lambda}(s-t) \sigma \sqrt{V_t} d \tilde B_t ds \\
	= &\big[ - 2 r_t + \theta^2 V_t - \frac{(1-2\rho^2) \sigma^2}{2} \psi^2(T-t) V_t \big] dt \\
	& - \theta^2 \int^T_t \sigma \frac{1}{\lambda} R_{\lambda}(s-t) ds \sqrt{V_t} d \tilde B_t +  \frac{(1-2\rho^2) \sigma^2}{2} \int^T_t  \sigma \psi^2(T-s) \frac{1}{\lambda} R_{\lambda}(s-t) ds \sqrt{V_t} d \tilde B_t \\
	= &\big[ - 2 r_t + \theta^2 V_t - \frac{(1-2\rho^2) \sigma^2}{2} \psi^2(T-t) V_t \big] dt \\
	& + d \tilde B_t \cdot \sigma \sqrt{V_t}   \int^T_t \Big[ \frac{(1-2\rho^2) \sigma^2}{2} \psi^2(T-s) - \theta^2  \Big] \frac{1}{\lambda} R_{\lambda}(s-t) ds .
	\end{align*}
	The second equality is guaranteed by the stochastic Fubini theorem \cite{veraar2012fubini}.
	
	We claim the following representation for (\ref{Eq:U1simp})-(\ref{Eq:U2simp}).
	\begin{align}
	U_{1t} & = \sigma \rho M_t \sqrt{V_t}  \int^T_t \Big[ \frac{(1-2\rho^2) \sigma^2}{2} \psi^2(T-s) - \theta^2  \Big] \frac{1}{\lambda} R_{\lambda}(s-t) ds, \label{Eq:U1}\\
	U_{2t} & = \sigma \sqrt{1 - \rho^2} M_t \sqrt{V_t}  \int^T_t \Big[ \frac{(1-2\rho^2) \sigma^2}{2} \psi^2(T-s) - \theta^2  \Big] \frac{1}{\lambda} R_{\lambda}(s-t) ds. \label{Eq:U2}
	\end{align}
	Indeed, we only have to show
	\begin{equation}\label{Eq:Uequivalent}
	\int^T_t \Big[ \frac{(1-2\rho^2) \sigma^2}{2} \psi^2(T-s) - \theta^2  \Big] \frac{1}{\lambda} R_{\lambda}(s-t) ds = \psi (T-t). 
	\end{equation}
	Although one can verify \eqref{Eq:Uequivalent} in the same fashion as \cite[Lemma 4.4]{abi2017affine}, we still detail the derivation here for a self-contained paper. As 
	\begin{align*}
	& \int^T_t \Big[ \frac{(1-2\rho^2) \sigma^2}{2} \psi^2(T-s) - \theta^2  \Big] \frac{1}{\lambda} R_{\lambda}(s-t) ds \\
	= & \int^{T - t}_0 \Big[ \frac{(1-2\rho^2) \sigma^2}{2} \psi^2(T- t - s) - \theta^2  \Big] \frac{1}{\lambda} R_{\lambda}(s) ds \\
	= & \big[ \frac{(1-2\rho^2) \sigma^2}{2} \psi^2 - \theta^2  \big] * \frac{1}{\lambda} R_{\lambda}(T - t),
	\end{align*}
	we have
	\begin{align*}
	& \int^T_t \Big[ \frac{(1-2\rho^2) \sigma^2}{2} \psi^2(T-s) - \theta^2  \Big] \frac{1}{\lambda} R_{\lambda}(s-t) ds - \psi (T-t)\\
	= & \big[ \frac{(1-2\rho^2) \sigma^2}{2} \psi^2 - \theta^2  \big] * \frac{1}{\lambda} R_{\lambda}(T - t) - K* \big[ \frac{(1-2\rho^2) \sigma^2}{2} \psi^2 - \lambda \psi - \theta^2  \big](T-t) \\
	= & \big[ \frac{(1-2\rho^2) \sigma^2}{2} \psi^2 - \theta^2  \big] * \big[ \frac{1}{\lambda} R_{\lambda} - K \big](T - t) + \lambda K*\psi(T-t)\\
	= & - R_\lambda*K*\big[ \frac{(1-2\rho^2) \sigma^2}{2} \psi^2 - \theta^2  \big](T-t) + \lambda K*\psi(T-t) .
	\end{align*}
	The application of (\ref{Eq:psi}) leads to
	\begin{equation}
	R_\lambda * \psi = R_\lambda*K*\big[ \frac{(1-2\rho^2) \sigma^2}{2} \psi^2 - \lambda \psi - \theta^2  \big].
	\end{equation}
	Consequently,
	\begin{align*}
	& - R_\lambda*K*\big[ \frac{(1-2\rho^2) \sigma^2}{2} \psi^2 - \theta^2  \big](T-t) + \lambda K*\psi(T-t) \\
	& = \big[ \lambda K - R_\lambda - \lambda K* R_\lambda \big]* \psi(T-t) = 0.
	\end{align*}
	This shows that
	\begin{align}
	dZ_t = &\big[ - 2 r_t + \theta^2 V_t - \frac{(1-2\rho^2) \sigma^2}{2} \psi^2(T-t) V_t \big] dt + \frac{U_{1t}}{M_t} d\tilde W_{1t} + \frac{U_{2t}}{M_t} d W_{2t}.
	\end{align}
	Applying It\^o's lemma to $M_t= 2e^{Z_t}$ with function $f(z) = 2 e^z$ yields
	\begin{align*}
	d M_t = & M_t d Z_t + \frac{1}{2} M_t dZ_t dZ_t \\
	= & M_t \big[ - 2 r_t + \theta^2 V_t - \frac{(1-2\rho^2) \sigma^2}{2} \psi^2(T-t) V_t \big] dt + \frac{U^2_{1t} + U^2_{2t}}{2 M_t} dt \\
	& + U_{1t} d\tilde W_{1t} + U_{2t} d W_{2t} \\
	= & \big[ -2 r_t + \theta^2 V_t \big] M_t dt + \big[ 2 \theta \sqrt{V_t} U_{1t} + \frac{U^2_{1t}}{M_t} \big] dt + U_{1t} dW_{1t} + U_{2t} dW_{2t}.
	\end{align*}
	
	\noindent {\bf Property (3)}.
	
	The proof for the property of $Y_t$ in \cite[Theorem 4.3]{abi2017affine} indicates
	\begin{align*}
	& \int^T_0 \big[ - \theta^2 \xi_0(s) + \frac{(1-2\rho^2) \sigma^2}{2} \psi^2(T-s) \xi_0(s) \big] ds \\
	& = \int^T_0 \big[ - \theta^2 V_0 + (\kappa \phi - \lambda V_0) \psi(s) + \frac{(1-2\rho^2)\sigma^2}{2} \psi^2(s) V_0 \big] ds. 
	\end{align*}
	
	Under the fractional kernel, we show by integration by parts that
	\begin{equation}
	\int^T_0 \big[ - \theta^2 - \lambda \psi(s) + \frac{(1-2\rho^2)\sigma^2}{2} \psi^2(s) \big] ds = I^{1 - \alpha} \psi(T).
	\end{equation}
	This gives the desired result.
	
	\noindent	{\bf Property (4)}.
	
	It is sufficient to consider the case with $p > 2$. As $\psi(t)$ is continuous on $[0, T]$ and $M_t$ is essentially bounded, 
	\begin{align*}
	\E\Big[ \big(\int^T_0 U^2_{it} dt \big)^{p/2} \Big] \leq C \E\Big[ \big(\int^T_0 V_t dt \big)^{p/2} \Big] \leq C \int^T_0 \E \big[ V^{p/2}_t \big] dt  \leq C\sup_{t \in [0, T]} \E\big[  V^{p/2}_t \big] < \infty.
	\end{align*}	
	The last term is finite by \cite[Lemma 3.1]{abi2017affine}.
\end{proof}

We first propose a candidate optimal control $u^*$. In the following theorem, we prove the admissibility of $u^*$ and the integrability of the corresponding $X^*$. Theorem \ref{Thm:X^*u^*} is in the spirit of \cite{lim2002complete,lim2004incomplete,shen2015mean}. Finally, we prove the optimality of $u^*$ in \eqref{Eq:u*} by Theorem \ref{Thm:Sol}.

\begin{theorem}\label{Thm:X^*u^*}
	Assume Assumption \ref{Assum:K} holds and (\ref{Eq:psi}) has a unique continuous solution on $[0, T]$. Denote $A_t \triangleq \theta + \rho \sigma \psi(T-t)$. Suppose Assumption \ref{Assum:V} holds with constant $a$ given the following:
	\begin{equation}\label{Eq:Const_a}
	a = \max \Big\{ 2 p |\theta| \sup_{t \in [0, T]} | A_t |, (8p^2 - 2p) \sup_{t \in [0, T]}  A^2_t \Big\}, \quad \text{for certain } p > 2. 
	\end{equation}
	Consider 
	\begin{equation}\label{Eq:u*}
	u^*(t) = (\theta + \rho \sigma \psi (T-t)) \sqrt{V_t} (\zeta^* e^{-\int^T_t r_s ds} - X^*_t),
	\end{equation}
	where $X^*_t$ is the wealth process under $u^*$ and $\zeta^* = c - \eta^*$ with
	\begin{equation}\label{Eq:eta*}
	\eta^* = \frac{e^{-\int^T_0 r_s ds}M_0 x_0 - e^{-\int^T_0 2 r_s ds} M_0 c}{2 - e^{-\int^T_0 2 r_s ds} M_0}.
	\end{equation}
	
	$u^*(\cdot)$ in (\ref{Eq:u*}) is admissible and $X^*$ under $u^*(\cdot)$ satisfies 
	\begin{equation}\label{Eq:X*Integral}
	\E \Big[ \sup_{ t \in [0, T]} | X^*_t |^p \Big] < \infty,
	\end{equation}
	for $p \geq 1$. Moreover,
	\begin{equation}\label{Eq:X*bound}
	\zeta^* e^{-\int^T_t r_s ds} - X^*_t \geq 0, \quad  \text{$\p$-$\as$}, \forall \; t \in [0, T].
	\end{equation}
\end{theorem}
\begin{proof}
	The wealth process under $u^*$ is given by
	\begin{equation}
	\left\{
	\begin{array}{rcl}
	dX^*_t &=& \big[ r_t X^*_t + \theta A_t V_t (\zeta^* e^{-\int^T_t r_s ds} - X^*_t) \big] dt + A_t \sqrt{V_t} (\zeta^* e^{-\int^T_t r_s ds} - X^*_t) dW_{1t},  \\
	X^*_0 &=& x_0.
	\end{array}\right.
	\end{equation}
	
	To find a solution to $X^*$, define $Y_t$ satisfying 
	\begin{equation}
	\left\{
	\begin{array}{rcl}
	dY_t &=& - r_t Y_t dt - \theta \sqrt{V_t} Y_t dW_{1t}  + Y_t \sqrt{1 - \rho^2} \sigma \psi(T-t) \sqrt{V_t} dW_{2t}, \\
	Y_0 &=& M_0 (\zeta^* e^{-\int^T_0 r_s ds}   -x_0 ).
	\end{array}\right.
	\end{equation}
	The unique solution of $Y_t$ is given by
	\begin{align*}
	Y_t =& Y_0 \exp \Big[ - \frac{1}{2} \int^t_0 \big( 2r_s + \theta^2 V_s + (1 - \rho^2) \sigma^2 \psi^2(T-s) V_s \big) ds - \int^t_0 \theta \sqrt{V_s} dW_{1s}  \\
	& \qquad \quad +  \int^t_0 \sqrt{1 - \rho^2} \sigma \psi(T-s) \sqrt{V_s} dW_{2s} \Big].
	\end{align*}
	It\^o's lemma yields
	\begin{equation}\label{Eq:X*2Y}
	X^*_t = \zeta^* e^{-\int^T_t r_s ds} - \frac{Y_t}{M_t}
	\end{equation}
	as the unique solution of the wealth process. Indeed,
	\begin{equation}
	d \frac{Y_t}{M_t} = \Big[ r_t \frac{Y_t}{M_t} - \theta A_t V_t \frac{Y_t}{M_t} \Big] dt - A_t \sqrt{V_t} \frac{Y_t}{M_t} dW_{1t}.
	\end{equation}
	The existence of $u^*$ is also guaranteed by the existence of the solution $X^*$. Furthermore, $\frac{Y_t}{M_t} = \frac{Y_0}{M_0} \Phi(t)$, where
	\begin{align*}
	\Phi(t) \triangleq & \exp \Big[ \int^t_0 \big[ r_s - \big( \theta A_s +  \frac{A^2_s}{2} \big) V_s \big] ds - \int^t_0 A_s \sqrt{V_s} dW_{1s} \Big].
	\end{align*}
	As $Y_t/M_t \geq 0$, (\ref{Eq:X*bound}) follows from (\ref{Eq:X*2Y}).
	
	For (\ref{Eq:X*Integral}), note that by Doob's maximal inequality and \cite[Lemma 7.3]{abi2017affine},
	\begin{align*}
	& \E \Big[ \sup_{ t \in [0, T]} | \Phi(t) |^p \Big] \\
	& \leq C \E \Big[ \sup_{ t \in [0, T]} \Big| e^{- \int^t_0 \theta A_s V_s ds} \Big|^{2p} \Big] + C \E \Big[ \sup_{ t \in [0, T]} \Big|  \exp \Big(- \int^t_0 \frac{A^2_s}{2}  V_s ds - \int^t_0 A_s \sqrt{V_s} dW_{1s} \Big) \Big|^{2p}\Big] \\
	& \leq C \E \Big[ e^{2p \int^T_0 |\theta A_s| V_s ds} \Big] + C \E \Big[ \exp \Big(- \int^T_0 p A^2_s  V_s ds - \int^T_0 2p A_s \sqrt{V_s} dW_{1s} \Big) \Big].
	\end{align*}
	The first term is finite by Assumption \ref{Assum:V} with constant $a = 2 p |\theta| \sup_{t \in [0, T]} | A_t |$. The second term is also finite. In fact, by H\"older's inequality and Assumption \ref{Assum:V} with a constant $a = (8p^2 - 2p) \sup_{t \in [0, T]}  A^2_t$,
	\begin{align*}
	& \E \Big[ \exp \Big(- \int^T_0 p A^2_s  V_s ds - \int^T_0 2p A_s \sqrt{V_s} dW_{1s} \Big) \Big] \\
	& \leq \Big\{ \E \Big[ e^{(8p^2 - 2p) \int^T_0 A^2_s  V_s ds} \Big] \Big\}^{1/2} \Big\{ \E \Big[ \exp \Big(- 8 p^2 \int^T_0 A^2_s  V_s ds - 4 p \int^T_0 A_s \sqrt{V_s} dW_{1s} \Big) \Big] \Big\}^{1/2} \\
	& < \infty.
	\end{align*}
	$\E \Big[ \sup_{ t \in [0, T]} | X^*_t |^p \Big] < \infty$ is proved.  As for admissibility of $u^*$, $u^*$ is $\F$-adapted at first. For integrability, let $1/\hat{p} + 1/\hat{q} = 1$, $\hat{p},\, \hat{q} >1$, we have
	\begin{align*}
	& \E\Big[ \Big(\int^T_0 |\sqrt{V_t} u^*_t |dt \Big)^2 \Big] \leq C \E\Big[ \Big(\int^T_0 |A_t V_t \Phi(t) |dt \Big)^2 \Big] \\
	& \leq C \E\Big[ \sup_{ t \in [0, T]}  \Phi^2(t) \Big(\int^T_0 V_t dt \Big)^2 \Big] \leq C \Big\{\E\Big[ \sup_{ t \in [0, T]}  \Phi^{2\hat{p}} (t) \Big] \Big\}^{1/\hat{p}} \Big\{\E\Big[ \Big(\int^T_0 V_t dt \Big)^{2 \hat{q}} \Big] \Big\}^{1/\hat{q}} \\
	& \leq C \Big\{ \E\Big[ \sup_{ t \in [0, T]}  \Phi^{2\hat{p}} (t) \Big] \Big\}^{1/\hat{p}} \Big( \sup_{t \in [0, T]} \E\big[ \ V_t^{2 \hat{q}} \big] \Big)^{1/\hat{q}} < \infty
	\end{align*}
	and 
	\begin{align*}
	& \E\Big[ \int^T_0 |u^*_t|^2 dt \Big] \leq C \E\Big[ \int^T_0 A^2_t V_t \Phi^2(t) dt \Big] \\
	& \leq C \E\Big[ \sup_{ t \in [0, T]}  \Phi^2(t) \int^T_0 V_t dt \Big] \leq C \Big\{ \E\Big[ \sup_{ t \in [0, T]}  \Phi^{2\hat{p}} (t) \Big] \Big\}^{1/\hat{p}} \Big\{\E\Big[ \Big(\int^T_0 V_t dt \Big)^{\hat{q}} \Big] \Big\}^{1/\hat{q}} \\
	& \leq C \Big\{ \E\Big[ \sup_{ t \in [0, T]}  \Phi^{2\hat{p}} (t) \Big] \Big\}^{1/\hat{p}} \Big( \sup_{t \in [0, T]} \E\big[ \ V_t^{\hat{q}} \big] \Big)^{1/\hat{q}} < \infty.
	\end{align*}
	The last terms in the two inequalities above are finite by \cite[Lemma 3.1]{abi2017affine} and take $p = 2 \hat{p}$. 
\end{proof}

We are now ready to prove $u^*$ in \eqref{Eq:u*} is optimal and to derive the efficient frontier.
\begin{theorem}\label{Thm:Sol}
	Suppose the assumptions in Theorem \ref{Thm:X^*u^*} hold, then the optimal investment strategy for Problem (\ref{Eq:obj}) is given by \eqref{Eq:u*}. Moreover, \eqref{Eq:u*} is unique under a given solution $(S, V, W_1, W_2)$ to (\ref{vol})-(\ref{stock}). The variance of $X^*_T$ is
	\begin{equation}\label{Eq:VarX*}
	\text{Var} [X^*_T] = \frac{M_0}{ 2 - e^{-\int^T_0 2 r_s ds} M_0} \big( c e^{-\int^T_0 r_s ds} - x_0 \big)^2.
	\end{equation}	
\end{theorem}
\begin{proof}
	First, we consider the inner Problem (\ref{Eq:innerobj}) with an arbitrary $\zeta \in \R$. Denote $h_t = \zeta e^{-\int^T_t r_s ds}$. By It\^o's lemma with the property of $M$ and completing the square, for any admissible strategy $u$,
	\begin{align*}
	& d \frac{1}{2} M_t (X_t - h_t)^2 \\
	& =  \frac{1}{2} \big[ (X_t - h_t)^2 M_t \theta^2 V_t + 2 (X_t - h_t)^2 \theta \sqrt{V_t} U_{1t} + (X_t - h_t)^2 \frac{U^2_{1t}}{M_t}  + 2 M_t (X_t - h_t) \theta \sqrt{V_t} u_t \\
	& \qquad + 2 ( X_t - h_t) u_t U_{1t} + M_t u^2_t \big] dt \\
	& \quad + \frac{1}{2} \big[ (X_t - h_t)^2 U_{1t} + 2 M_t (X_t - h_t) u_t \big] dW_{1t} + \frac{1}{2} (X_t - h_t)^2 U_{2t} dW_{2t} \\
	& = \frac{1}{2} M_t \Big[ u_t + \big(\theta \sqrt{V_t} + \frac{U_{1t}}{M_t}\big)(X_t - h_t) \Big]^2 dt \\
	& \quad + \frac{1}{2} \big[ (X_t - h_t)^2 U_{1t} + 2 M_t (X_t - h_t) u_t \big] dW_{1t} + \frac{1}{2} (X_t - h_t)^2 U_{2t} dW_{2t}.
	\end{align*}
	
	As $M_t$ and $h_t$ are bounded, $\E\Big[ \int^T_0 U^2_{it} dt \Big] < \infty$ for $i =1, 2$, $u_t$ is admissible, and $X_t$ has $\p$-$\as$ continuous paths, then stochastic integrals
	\begin{equation*}
	\int^t_0 \big[ (X_s - h_s)^2 U_{1s} + 2 M_s (X_s - h_s) u_s \big] dW_{1s} \quad \text{and} \quad \int^t_0 (X_s - h_s)^2 U_{2s} dW_{2s}
	\end{equation*}
	are $(\F, \p)$-local martingales. There is an increasing localizing sequence of stopping times $\{ \tau_k \}_{k = 1, 2, ...}$ such that $\tau_k \uparrow T$ when $ k \rightarrow \infty$. The local martingales stopped by $\{ \tau_k \}_{k = 1, 2, ...}$ are true martingales. Consequently,
	\begin{equation}
	\frac{1}{2} \E[M_{\tau_k} (X_{\tau_k} - h_{\tau_k})^2 ] = \frac{1}{2} M_0 (x_0 - h_0)^2 + \frac{1}{2} \E \Big[ \int^{\tau_k}_0 M_t \Big( u_t + \big(\theta \sqrt{V_t} + \frac{U_{1t}}{M_t}\big)(X_t - h_t) \Big)^2 dt \Big].
	\end{equation}
	
	From (\ref{Eq:wealth}), by Doob's maximal inequality and the admissibility of $u(\cdot)$, 
	\begin{equation}
	\E[X^2_{\tau_k}] \leq C \Big[ x^2_0 + \E \Big[ \big( \int^T_0 | u_t \sqrt{V_t}| dt \big)^2 \Big] + \E \Big[  \int^T_0 u^2_t dt  \Big] \Big] < \infty.
	\end{equation}
	Then $M_{\tau_k} (X_{\tau_k} - h_{\tau_k})^2$ is dominated by a non-negative integrable random variable for all $k$. Sending $k$ to infinity, by the dominated convergence theorem and the monotone convergence theorem, we derive
	\begin{equation}\label{Eq:Square}
	\E[ (X_T - \zeta )^2 ] = \frac{1}{2} M_0 (x_0 - h_0)^2 + \frac{1}{2} \E \Big[ \int^T_0 M_t \Big( u_t + \big(\theta \sqrt{V_t} + \frac{U_{1t}}{M_t}\big)(X_t - h_t) \Big)^2 dt \Big].
	\end{equation}
	Therefore, the cost functional $\E[ (X_T - \zeta )^2 ]$ is minimized when
	\begin{equation}
	u_t = - \big(\theta \sqrt{V_t} + \frac{U_{1t}}{M_t}\big)(X_t - h_t).
	\end{equation}
	Then $ \E[ (X_T - \zeta )^2 ] = \frac{1}{2} M_0 (x_0 - h_0)^2$. The uniqueness of $u^*$ follows directly from \eqref{Eq:Square} and $M_t > 0$, $\p$-$\as$, $\forall \; t \in [0, T]$. To solve the outer maximization problem in (\ref{Eq:maxminobj}), consider
	\begin{equation}
	J(x_0; u(\cdot)) = \frac{1}{2} M_0 \big[x_0 - (c - \eta) e^{- \int^T_0 r_s ds} \big]^2 - \eta^2.
	\end{equation}
	The first and second order derivatives are
	\begin{align*}
	\frac{\partial J}{\partial \eta} &=  M_0 \big[x_0 - (c - \eta) e^{- \int^T_0 r_s ds} \big] e^{- \int^T_0 r_s ds} - 2 \eta, \\
	\frac{\partial^2 J}{\partial \eta^2} &=  M_0 e^{- 2 \int^T_0 r_s ds} - 2 < 0,
	\end{align*}
	where we have used the strict inequality $M_0 < 2 e^{\int^T_0 2 r_s ds}$, by Theorem \ref{Thm:M}.
	
	Then the optimal value for $\eta$ is given by (\ref{Eq:eta*}), solved from $\frac{\partial J}{\partial \eta} = 0$. $\text{Var}[X^*_T]$ is obtained by direct simplification of $J(x_0; u(\cdot))$ with $\eta^*$.	
\end{proof}

Although the Volterra Heston model is non-Markovian and non-semimartingale in nature, the optimal control $u^*$ in \eqref{Eq:u*} does not rely on the whole volatility path. Moreover, the optimal amount of wealth in the stock, $\pi^*_t$, does not depend on the volatility value directly, but rather on the roughness and dynamics of volatility through parameters and the Riccati-Volterra equation \eqref{Eq:psi}. If we let kernel $K = \id$, it is then clear that the Volterra Heston model (\ref{vol}) reduces to the classic Heston model \cite{heston1993closed}. Our results in Theorem \ref{Thm:X^*u^*} and Theorem  \ref{Thm:Sol} indicate that the $u^*$ in (\ref{Eq:u*}) is optimal even under a general filtration $\F$. It extends the corresponding result in \cite{cerny2008mean,shen2015square} where the filtration is chosen as the Brownian filtration. As a sanity check, the following corollary verifies that our solution reduces to the one under the Heston model.

\begin{corollary}
	Consider the Heston model, that is, the kernel $K = \id$. Suppose other assumptions in Theorem \ref{Thm:X^*u^*} hold, then the optimal strategy \eqref{Eq:u*} is the same as the one in \cite{cerny2008mean}.
\end{corollary}
\begin{proof}
	Without loss of generality, suppose $r_t = 0$, as in \cite{cerny2008mean}. We first match $M_t/2$ in (\ref{Eq:M}) with opportunity process $L_t$ in \cite[Equation (3.2)]{cerny2008mean}.
	
	Note the resolvent in (\ref{Eq:R_lambda}) reduces to $R_\lambda (t) = \lambda e^{-\lambda t}$ and the forward variance in (\ref{Eq:xi}) is 
	\begin{equation}
	\xi_t(s) = e^{-\lambda (s - t)} V_t + \frac{\kappa \phi}{\lambda} \left( 1 - e^{-\lambda(s - t)} \right).
	\end{equation}
	Therefore,
	\begin{equation*}
	\int^T_t \xi_t(s) ds = \frac{1 - e^{-\lambda(T-t)}}{\lambda}  V_t + \frac{\kappa \phi}{\lambda} \left( T - t -  \frac{1 - e^{-\lambda(T-t)}}{\lambda} \right)
	\end{equation*}
	and 
	\begin{equation*}
	\int^T_t \psi^2(T-s) \xi_t(s) ds =  V_t \int^T_t \psi^2(T - s) e^{-\lambda(s - t)} ds + \frac{\kappa \phi}{\lambda} \int^T_t \big[ 1 - e^{-\lambda (s - t)} \big] \psi^2(T - s) ds.
	\end{equation*}
	Then
	\begin{equation}
	\int^T_t \big[ - \theta^2 \xi_t(s) + \frac{(1-2\rho^2) \sigma^2}{2} \psi^2(T-s) \xi_t(s) \big] ds = w(T-t) V_t + y(T-t),
	\end{equation}
	where
	\begin{align*}
	w(T-t) &\triangleq \frac{(1 - 2\rho^2)\sigma^2}{2} \int^T_t \psi^2(T - s) e^{-\lambda(s - t)} ds - \theta^2 \frac{1 - e^{-\lambda(T-t)}}{\lambda}, \\
	y(T-t) &\triangleq \frac{(1 - 2\rho^2)\sigma^2}{2} \frac{\kappa \phi}{\lambda} \int^T_t \big[ 1 - e^{-\lambda (s - t)} \big] \psi^2(T - s) ds - \theta^2 \frac{\kappa \phi}{\lambda} \left(T - t -  \frac{1 - e^{-\lambda(T-t)}}{\lambda} \right).
	\end{align*}
	
	Replacing $t$ with $T - t$ and taking derivative on $t$ give
	\begin{align*}
	\dot{w}(t) &= \frac{(1 - 2\rho^2)\sigma^2}{2} \psi^2(t) - \lambda \frac{(1 - 2\rho^2)\sigma^2}{2} \int^T_{T-t} \psi^2(T - s) e^{-\lambda(s - T + t)} ds - \theta^2 e^{-\lambda t} \\
	& =  \frac{(1 - 2\rho^2)\sigma^2}{2} \psi^2(t) - \lambda w(t) - \theta^2.
	\end{align*} 
	Comparing with (\ref{Eq:psi}), we find $w(t) = \psi(t)$. Moreover,
	\begin{align*}
	\dot{y}(t) & = \frac{(1 - 2\rho^2)\sigma^2}{2} \frac{\kappa \phi}{\lambda} \int^T_{T-t} \lambda e^{-\lambda (s - T + t)} \psi^2(T - s) ds - \theta^2 \frac{\kappa \phi}{\lambda} \left( 1 - e^{-\lambda t} \right) \\
	& = \kappa \phi w(t).
	\end{align*}
	
	$y(t)$ and $w(t)$ satisfy the same ODEs as in \cite[Equations (A.1)-(A.4)]{cerny2008mean}, with our notations. Therefore, $M_t/2$ in (\ref{Eq:M}) reduces to  $L_t$ in \cite[Equation (3.2)]{cerny2008mean}.
	
	Consider the inner Problem (\ref{Eq:innerobj}). With a constant $H = \zeta$, terms in the optimal hedge $\varphi(x, H)$ \cite[p.476]{cerny2008mean} are reduced to
	\begin{align}
	\xi = 0, \quad a= (\theta + \psi(T-t) \rho \sigma)/ S_t, \quad V = \zeta, \quad \text{and} \quad x + \varphi(x, H) \cdot S = X^*_t. 
	\end{align}
	Then it is clear that the optimal strategies are the same.
\end{proof}

\section{Numerical studies}\label{Sec:Numerical}
In this section, we restrict ourself to the case with $K(t) = \frac{t^{\alpha-1}}{\Gamma(\alpha)}$, $\alpha \in (1/2, 1)$, for the rough Heston model in \cite{eleuch2019char}. $\alpha = 1$ recovers the classic Heston model. We examine the effect of $\alpha$ on the optimal investment strategy and efficient frontier. 

The first step is to solve the Riccati-Volterra equation (\ref{Eq:psi}) numerically. Following \cite{eleuch2019char}, we use the fractional Adams method in \cite{diethelm2002predictor,diethelm2004error}. The convergence of this numerical method is given in \cite{li2009adams}. Readers may refer to \cite[Section 5.1]{eleuch2019char} for more details about the procedure.

In Figure (\ref{Fig:psi}), $\psi$ decreases when $\alpha$ becomes smaller under certain specific parameters, close to the calibration result in \cite{eleuch2019char} with one extra risk premium parameter $\theta$. However, one cannot expect $\psi$ to be monotone in $\alpha$ in general (see Figure (\ref{Fig:psi_new})). Figures (\ref{Fig:psi})-(\ref{Fig:psi_new}) also confirm the claim that $\psi \leq 0$ when $ 1 - 2\rho^2 > 0$.

\begin{figure}[!h]
	\centering
	\begin{minipage}{0.5\textwidth}
		\centering
		\includegraphics[width=0.9\textwidth]{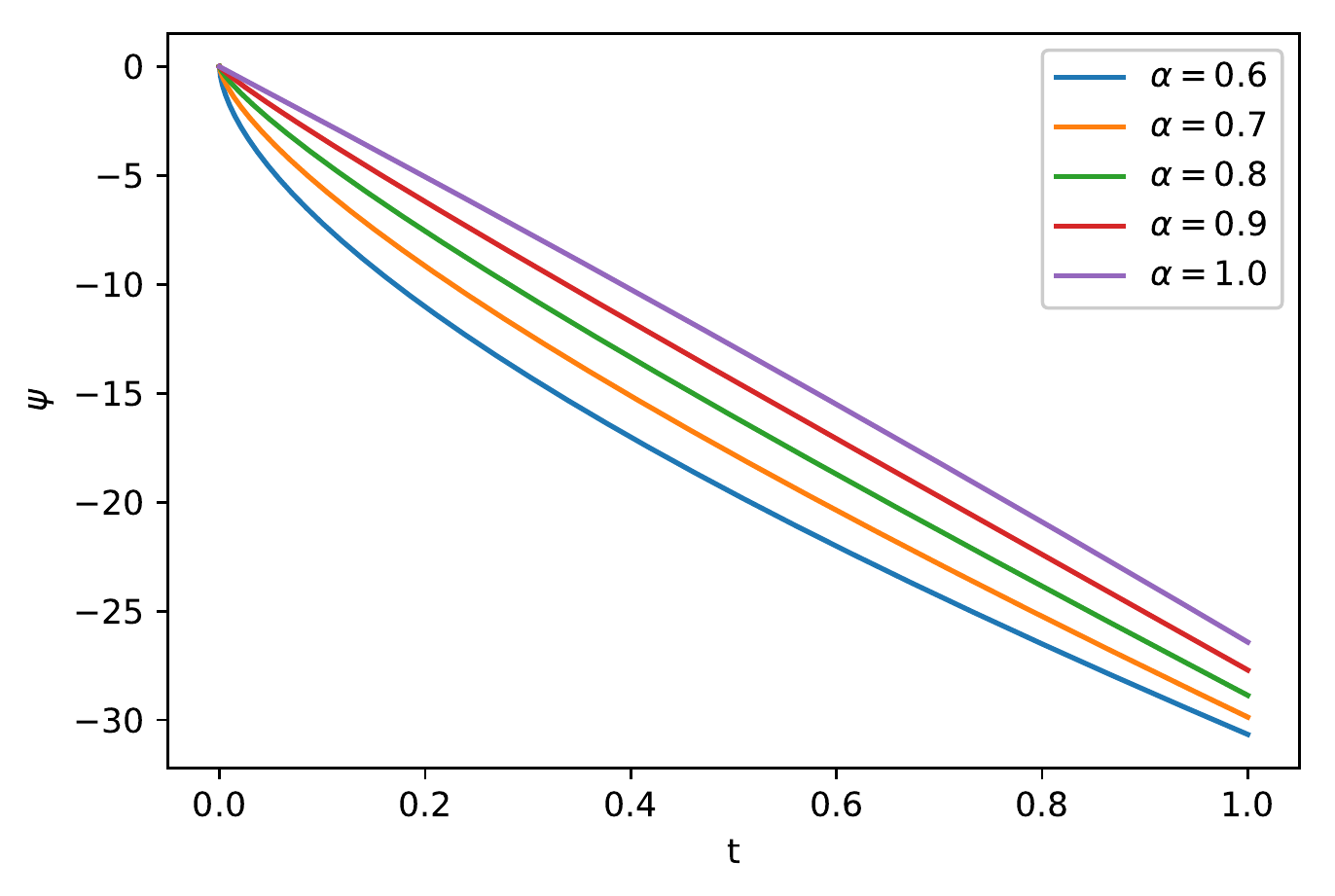}
		\subcaption{$\psi$ under parameters in \cite{eleuch2019char}}\label{Fig:psi}
	\end{minipage}%
	\begin{minipage}{0.5\textwidth}
		\centering
		\includegraphics[width=0.9\textwidth]{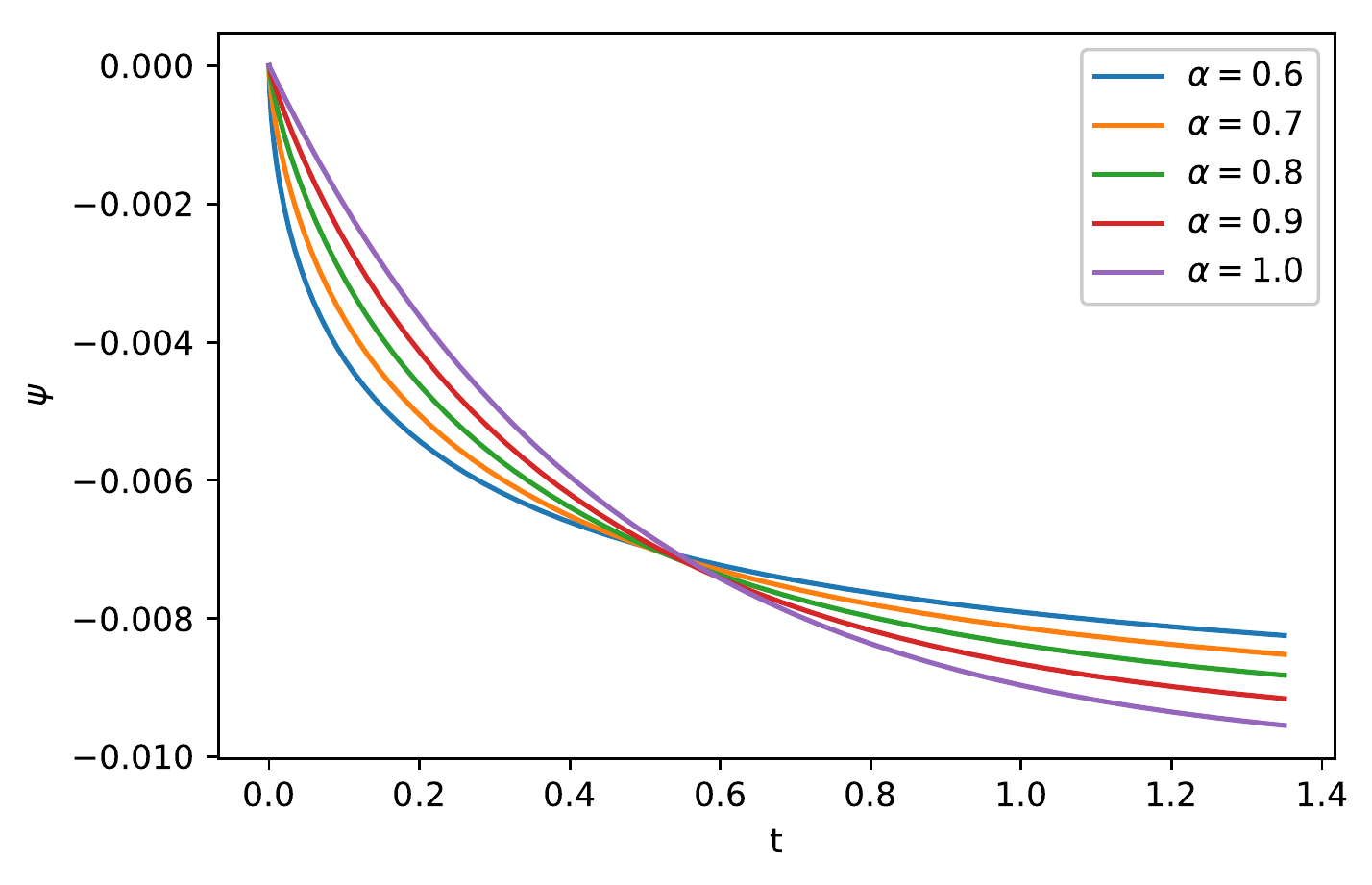}
		\subcaption{$\psi$ under another setting}\label{Fig:psi_new}
	\end{minipage}
	\caption{Plot of $\psi$ under different $\alpha$. Other parameters are as follows. In Figure (\ref{Fig:psi}), vol-of-vol $\sigma = 0.03$, mean-reversion speed $\kappa = 0.1$, risk premium parameter $\theta = 5$, correlation $\rho = -0.7$, and time horizon $T = 1$. In Figure (\ref{Fig:psi_new}), $\sigma = 0.04$, $\kappa = 2.25$, $\theta = 0.15$, $\rho = -0.56$, and $T=1.35$.}
\end{figure}

The relationship between $u^*$ and $\alpha$ is not straightforward and may change with different combinations of parameters. We emphasize that the following analysis is based on the parameter setting detailed in the descriptions of the figures. Consider the setting in Figure (\ref{Fig:psi}) first. Interestingly, the effect of $\alpha$ on $u^*$ is significantly influenced by $\sigma$. This can be explained using (\ref{Eq:u*}). If the correlation $\rho$ between stock and volatility is negative due to the leverage effect in the equity market, $\theta + \rho \sigma \psi (T-t)$ will increase as $\alpha$ decreases, as shown in Figure (\ref{Fig:psi}). In contrast, $\zeta^* e^{-\int^T_t r_s ds} - X^*_t \geq 0$ by Theorem \ref{Thm:X^*u^*}. Note 
\begin{equation}
\zeta^* = c - \eta^* = \frac{2c - e^{-\int^T_0 r_s ds}M_0 x_0}{2 - e^{-\int^T_0 2 r_s ds} M_0}.
\end{equation}
The $M_0$ in (\ref{Eq:fracM0}) is an increasing function on $\alpha$ because $\psi$ is negative. Then $\zeta^*$ will be smaller if $\alpha$ is smaller, under certain parameters. Therefore, $\zeta^* e^{-\int^T_t r_s ds} - X^*_t$ and $\theta + \rho \sigma \psi (T-t)$ move in different directions when $\alpha$ is decreasing. If $\sigma$ is small, $\zeta^* e^{-\int^T_t r_s ds} - X^*_t$ will dominate $\theta + \rho \sigma \psi (T-t)$. Then $u^*$ will decrease as $\alpha$ becomes smaller. If $\sigma$ is relatively large, $\theta + \rho \sigma \psi (T-t)$ will dominate $\zeta^* e^{-\int^T_t r_s ds} - X^*_t$. Then $u^*$ increases when $\alpha$ becomes smaller. The above effect of vol-of-vol $\sigma$ also appears under the parameters setting in Figure (\ref{Fig:psi_new}), where $\psi$ is not monotone in $\alpha$. Figures (\ref{Fig:u4smallsig})-(\ref{Fig:u4bigsig}) display the optimal investment strategy $u^*$. We make use of the open-source Python package {\bf differint}\footnote{Available at \url{https://github.com/differint/differint} } to calculate the fractional integrals $I^{1-\alpha}$ and $I^1$ in (\ref{Eq:fracM0}). Assumption \ref{Assum:V} is validated under the setting in Figures (\ref{Fig:u4smallsig})-(\ref{Fig:u4bigsig}).

\begin{figure}[!h]
	\centering
	\begin{minipage}{0.5\textwidth}
		\centering
		\includegraphics[width=0.9\textwidth]{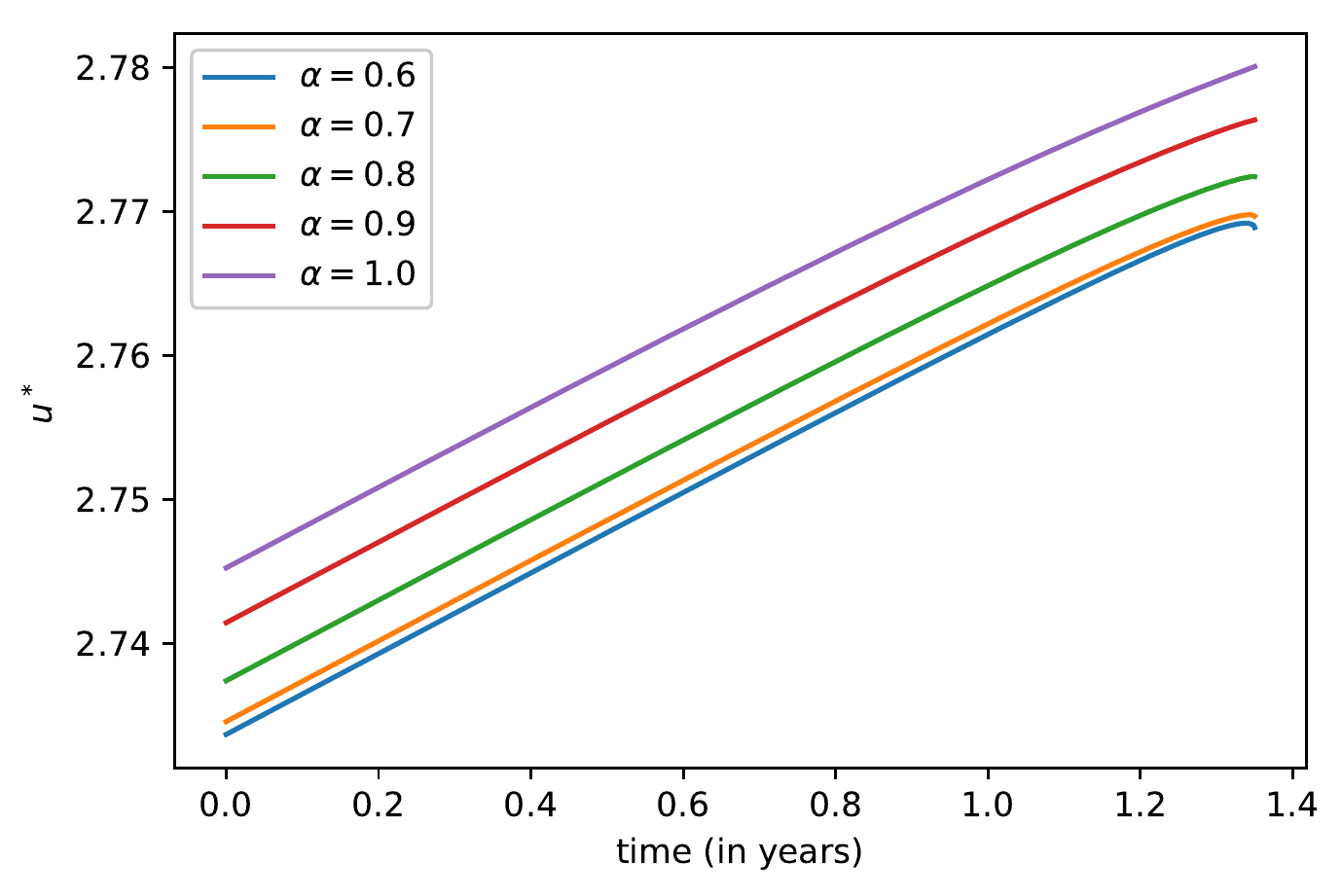}
		\subcaption{$u^*$ under $\sigma=0.04$}\label{Fig:u4smallsig}
	\end{minipage}%
	\begin{minipage}{0.5\textwidth}
		\centering
		\includegraphics[width=0.9\textwidth]{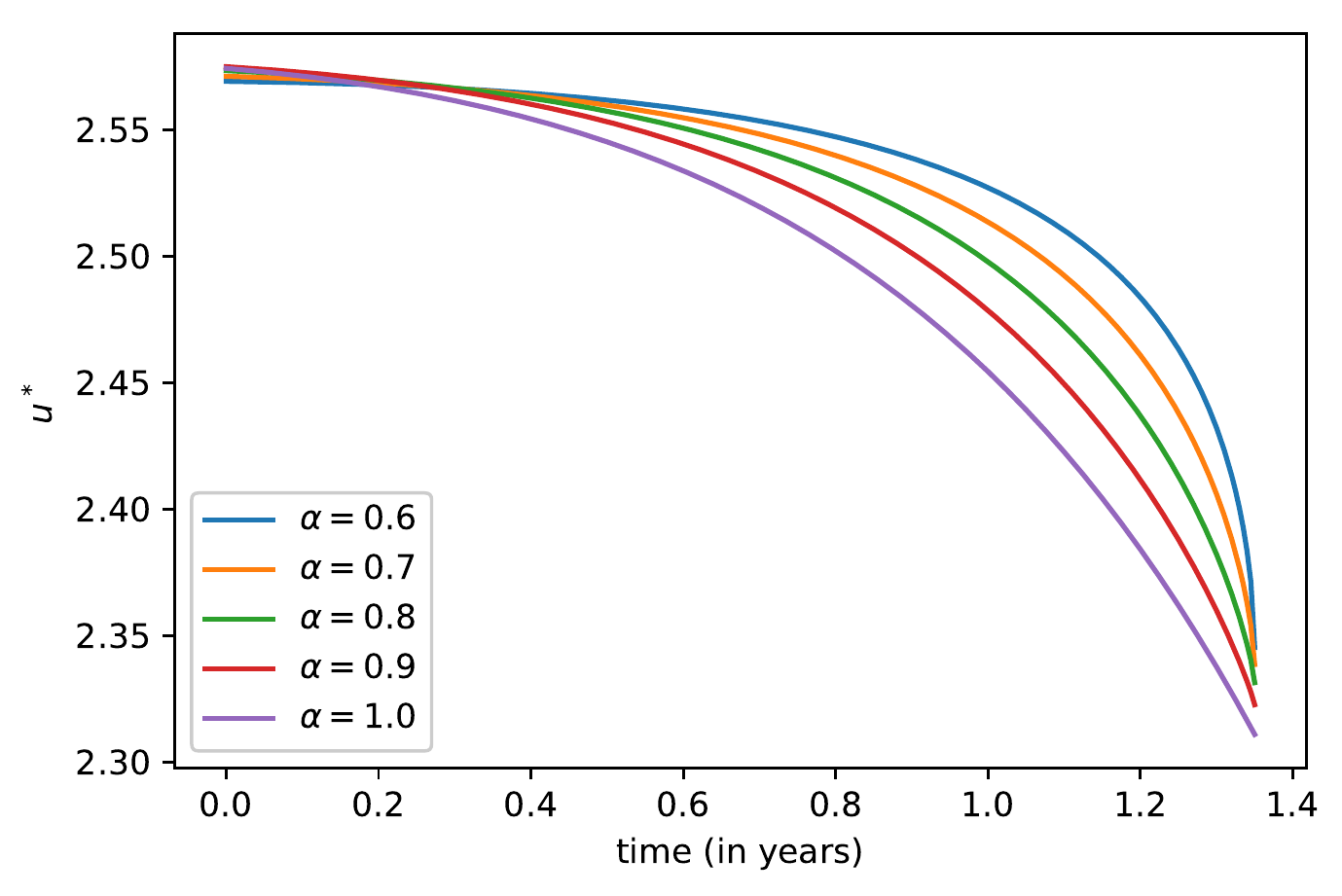}
		\subcaption{$u^*$ under $\sigma = 3$}\label{Fig:u4bigsig}
	\end{minipage}
	\caption{Optimal strategy $u^*$ with $\alpha=0.6, 0.7, 0.8, 0.9$, and $1.0$. In both subplots, we set initial wealth $x_0 = 1$, risk-free rate $r= 0.01$, initial variance $V_0 = 0.5$, long-term mean level $\phi = 0.04$, and expected terminal wealth $c = x_0 e^{(r+0.1)T}$. For simplicity, we set $V_t = 0.5$ and $X^*_t = 1$ for all time $t \in [0, T]$. The other parameters are the same as in Figure (\ref{Fig:psi_new}), namely, $\kappa = 2.25$,  $\theta = 0.15$, $\rho = -0.56$, and $T=1.35$. Figures (\ref{Fig:u4smallsig})-(\ref{Fig:u4bigsig}) only differ in the vol-of-vol $\sigma$. }
\end{figure}

Figures (\ref{Fig:u4smallsig})-(\ref{Fig:u4bigsig}) are a sensitivity analysis as we keep most of the parameters unchanged, and vary a few of them. Specifically, the use of constant $V_t$ and $X^*_t$ in Figures (\ref{Fig:u4smallsig})-(\ref{Fig:u4bigsig}) has the following interpretation. We are interested in the sensitivity of the optimal control on the Hurst parameter through $\alpha$. As the other parameters being fixed, if we observe $V_t = 0.5$ and $X^*_t = 1$ at $t \in [0, T]$, Figures (\ref{Fig:u4smallsig})-(\ref{Fig:u4bigsig}) illustrate the marginal effect of the Hurst parameter on the investment strategy. The constant values of $V_t$ and $X^*_t$ are not from a realized path.

Figures (\ref{Fig:u4smallsig})-(\ref{Fig:u4bigsig}) only provide a marginal effect of $\alpha$; thus, we conduct a further numerical analysis under the settings in \cite{abi2019lifting}. Consider a realistic situation in which the investor calibrates two sets of parameters for the Heston model and rough Heston model for a given implied volatility surface. We contrast the two strategies induced from the calibrated parameters. Figure (\ref{Fig:pi_sim}) exhibits the optimal amount of wealth $\pi^*$ with one simulation path of $V_t$ in Figure (\ref{Fig:roughv}) by the lifted Heston approach \cite{abi2019lifting}. Assumption \ref{Assum:V} holds true under the setting in Figure 3. Figure (\ref{Fig:At}) plots the $A_t = \theta + \rho \sigma \psi(T-t)$. Furthermore, $\zeta^* = 30.7458$ for the rough Heston model and $\zeta^* = 21.6351$ for the classic Heston model. The optimal strategy under the rough Heston model consistently suggests holding more in the stock. We stress that this is a persistent phenomenon for all of the simulation runs and is not limited to the particular one in Figure (\ref{Fig:pi_sim}). Indeed, Figures (\ref{Fig:uRough})-(\ref{Fig:uHn}) show the mean and confidence intervals of the strategies. The rough Heston strategy has larger values during the whole investment horizon. It can be explained with Figure (\ref{Fig:At}) and $\zeta^*$ reported. A rough Heston investor has a larger $A_t \zeta^*$ but a smaller $A_t$. Moreover, Figure (\ref{Fig:wealth}) illustrates that the rough Heston strategy has an average terminal wealth closer to the target $c = 1.1163$. Finally, we emphasize that Figure (\ref{Fig:pi_sim}) and Figures (\ref{Fig:uRough})-(\ref{Fig:uHn}) do not conflict with Figures (\ref{Fig:u4smallsig})-(\ref{Fig:u4bigsig}) because the mean-reversion rate and the vol-of-vol are different for the two strategies in Figure (\ref{Fig:pi_sim}). See \cite[Table 6]{abi2019lifting} and \cite[Table 4]{abi2019lifting} for more details.

\begin{figure}[!h]
	\centering
	\begin{minipage}{0.33\textwidth}
		\centering
		\includegraphics[width=0.95\textwidth]{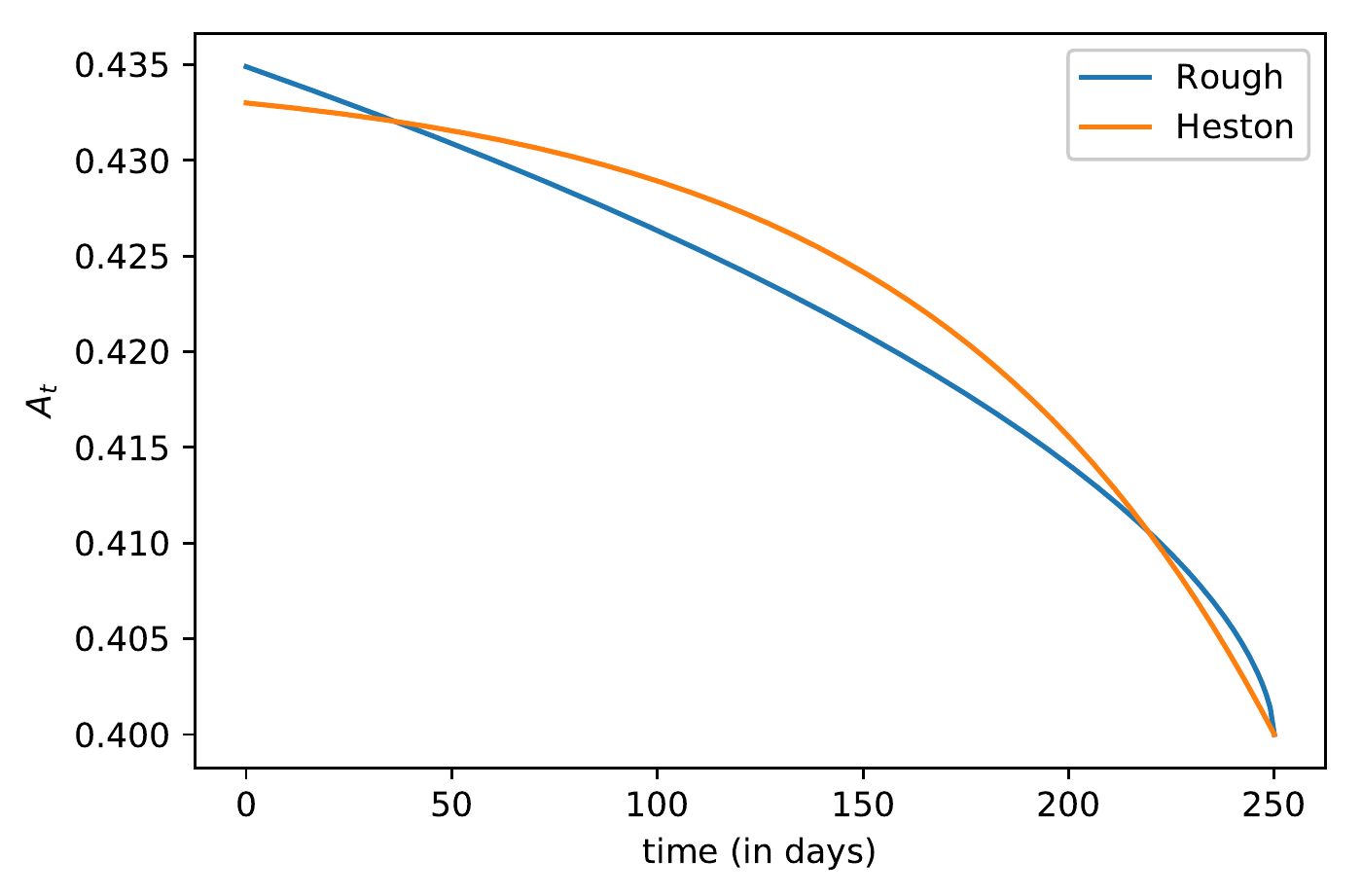}
		\subcaption{$A_t$}\label{Fig:At}
	\end{minipage}%
	\begin{minipage}{0.33\textwidth}
		\centering
		\includegraphics[width=0.95\textwidth]{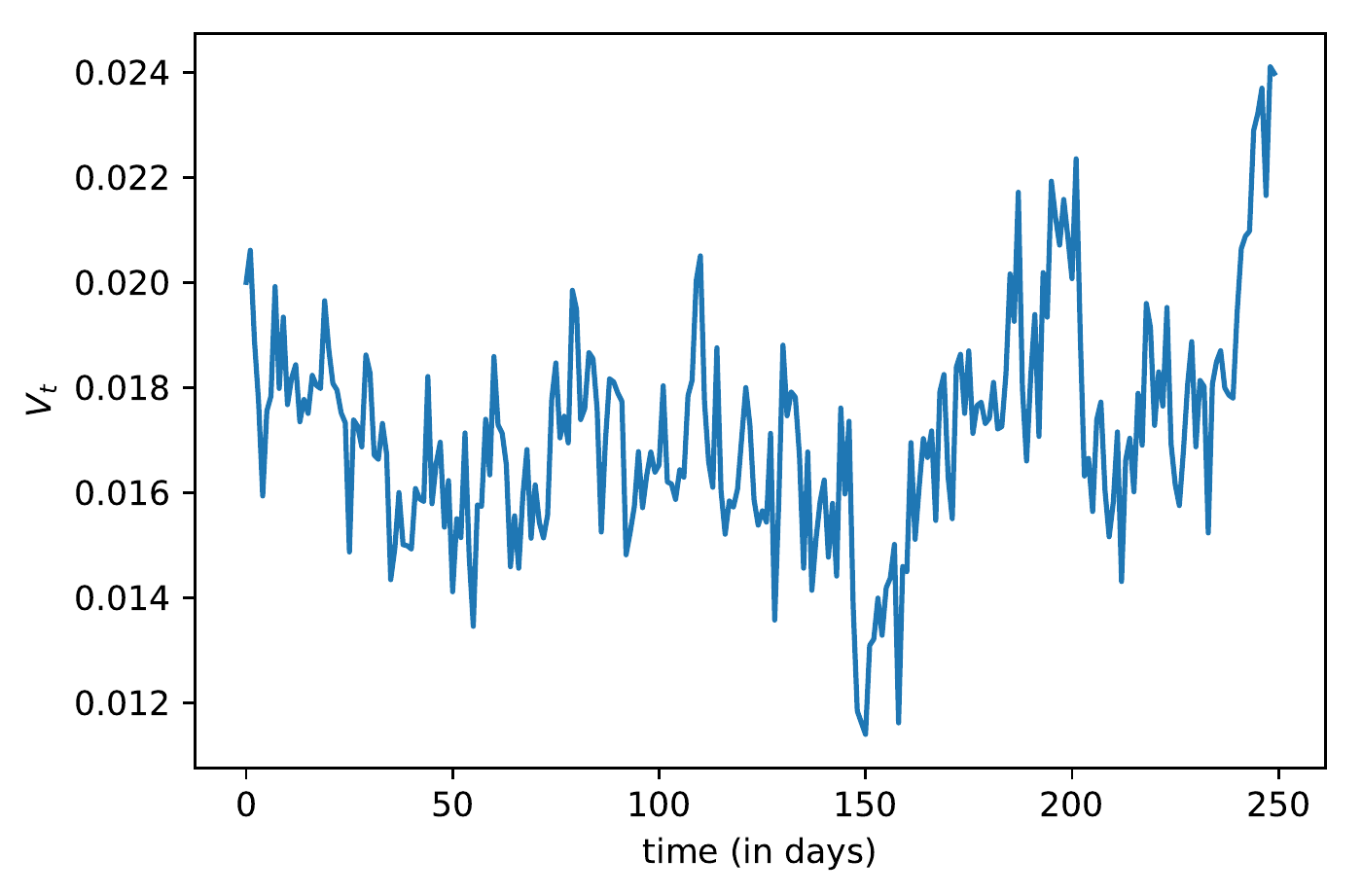}
		\subcaption{volatility}\label{Fig:roughv}
	\end{minipage}%
	\begin{minipage}{0.33\textwidth}
		\centering
		\includegraphics[width=0.95\textwidth]{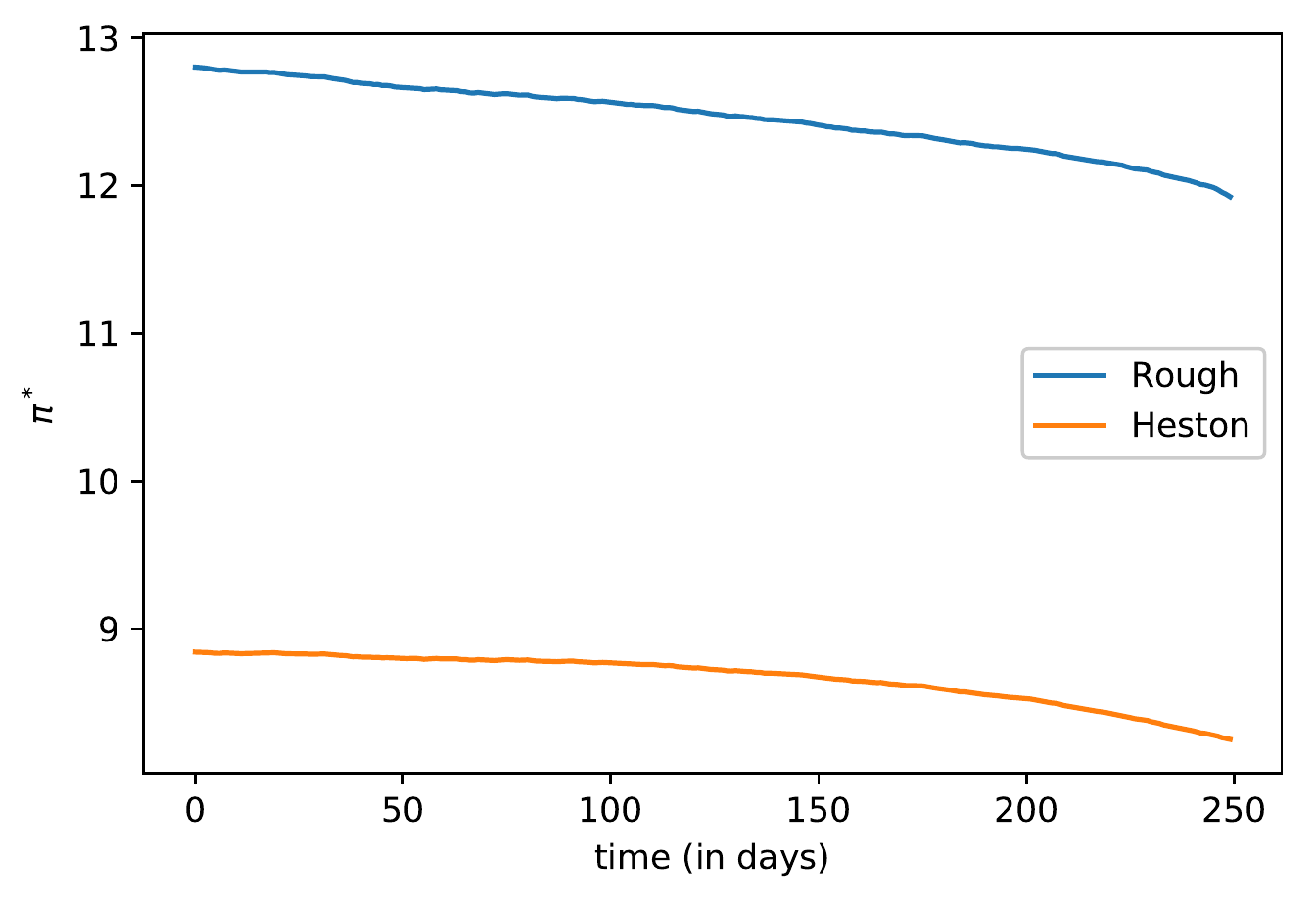}
		\subcaption{$\pi^*$}\label{Fig:pi_sim}
	\end{minipage}
	\caption{Investment strategies under the Heston and rough Heston models. The variance process is simulated with the lifted Heston model in \cite{abi2019lifting}. The parameters for simulation are specified in \cite[Equations (23) and (26)]{abi2019lifting} with $\alpha=0.6$. The path is rougher than that of the classic Heston model. Moreover, we implement the Euler scheme for the stock process. The simulation is run with $250$ time steps for one year, corresponding to the $250$ trading days in a year. The investor under the Heston model uses the calibrated parameters in \cite[Table 6]{abi2019lifting} to implement the optimal strategy with $\alpha = 0.59973346$ being the calibrated value. The investor under rough Heston model uses \cite[Table 4]{abi2019lifting} instead. We set $x_0 = 1$, $r = 0.01$, $\theta = 0.4$, $T = 1$, and $c = x_0 e^{(r+0.1)T} = 1.1163$. }
\end{figure}

\begin{figure}[!h]
	\centering
	\begin{minipage}{0.33\textwidth}
		\centering
		\includegraphics[width=0.9\textwidth]{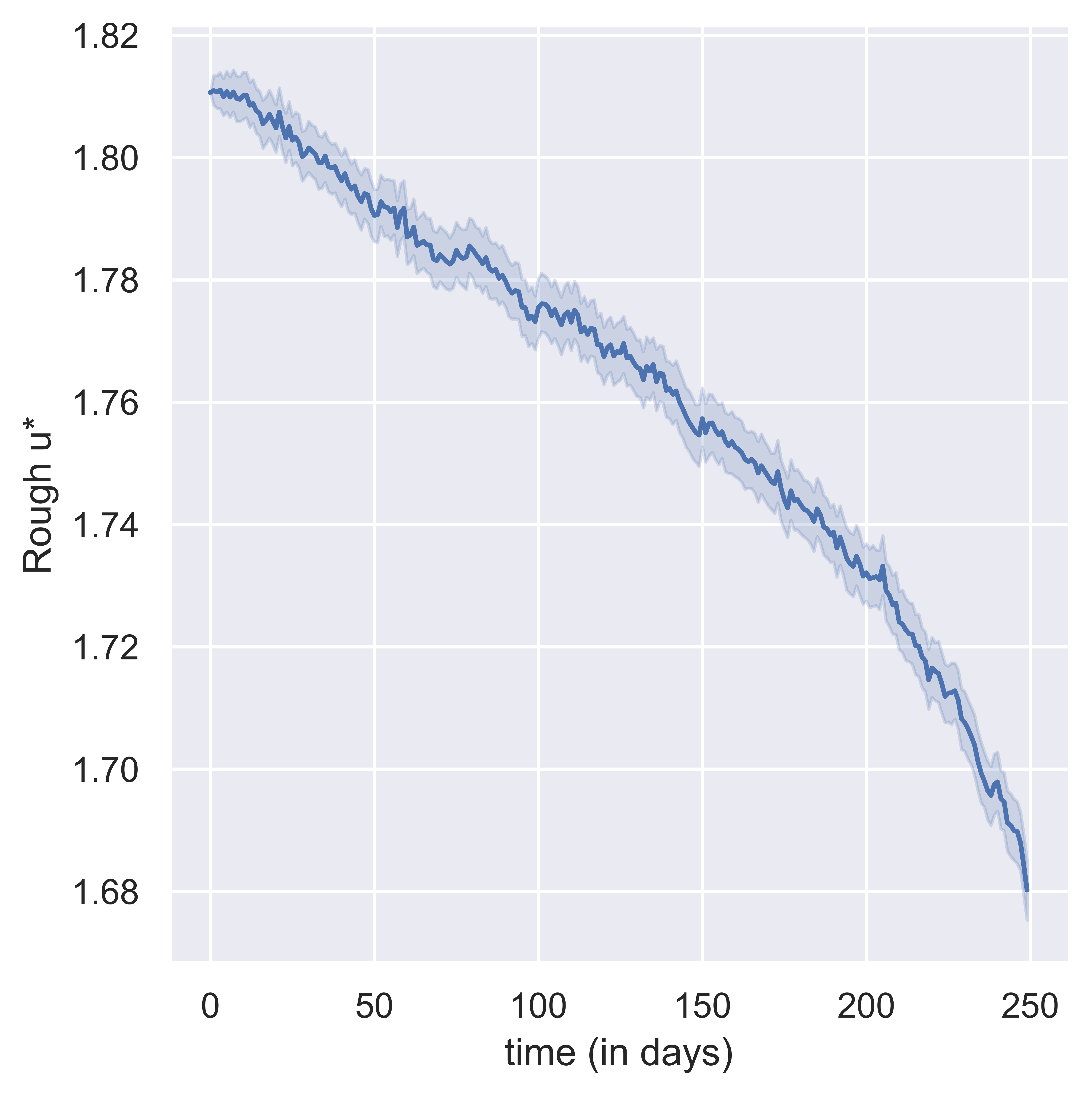}
		\subcaption{$u^*$ under the rough Heston model}\label{Fig:uRough}
	\end{minipage}%
	\begin{minipage}{0.33\textwidth}
		\centering
		\includegraphics[width=0.9\textwidth]{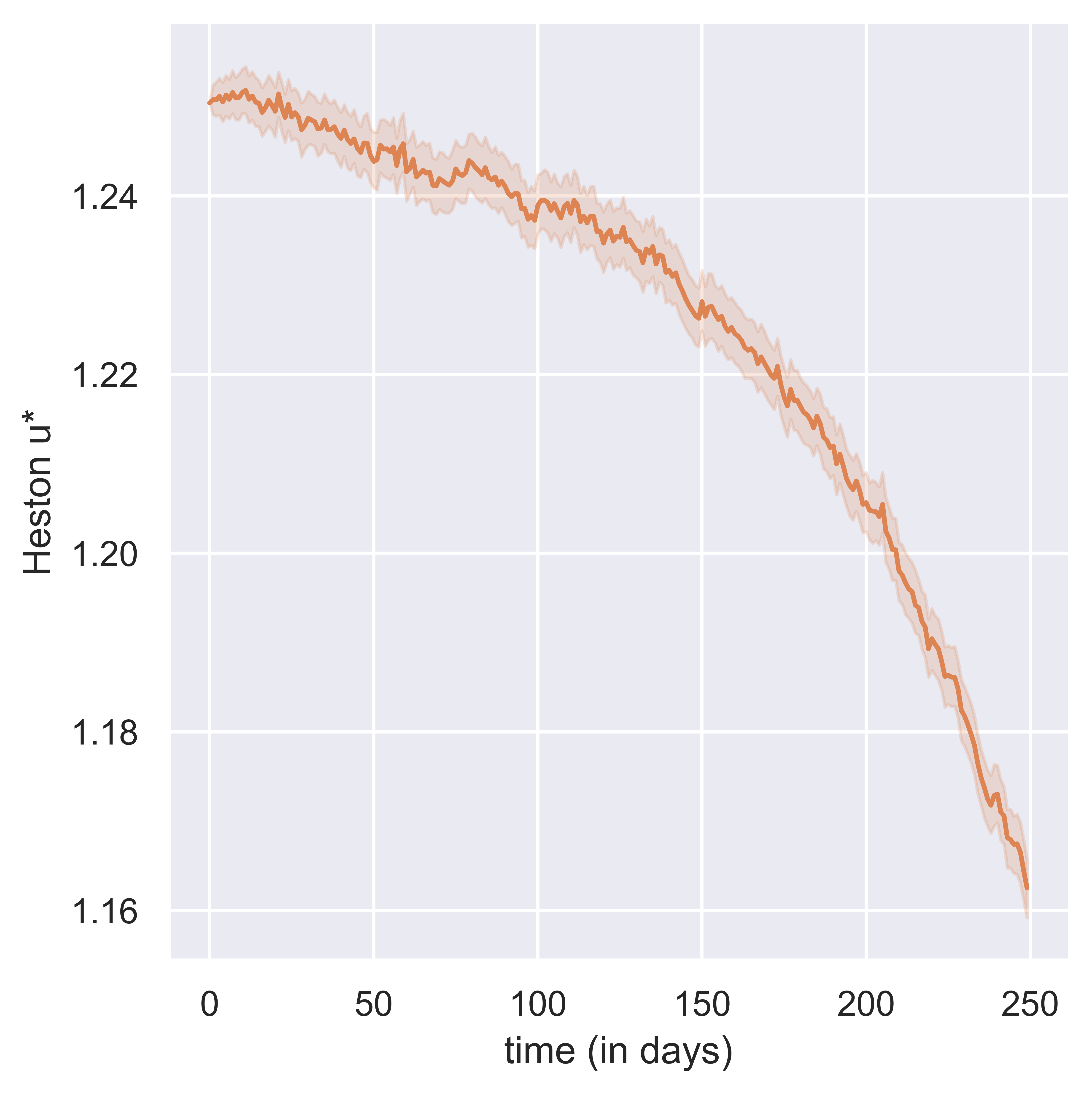}
		\subcaption{$u^*$ under the Heston model}\label{Fig:uHn}
	\end{minipage}%
	\begin{minipage}{0.33\textwidth}
		\centering
		\includegraphics[width=0.9\textwidth]{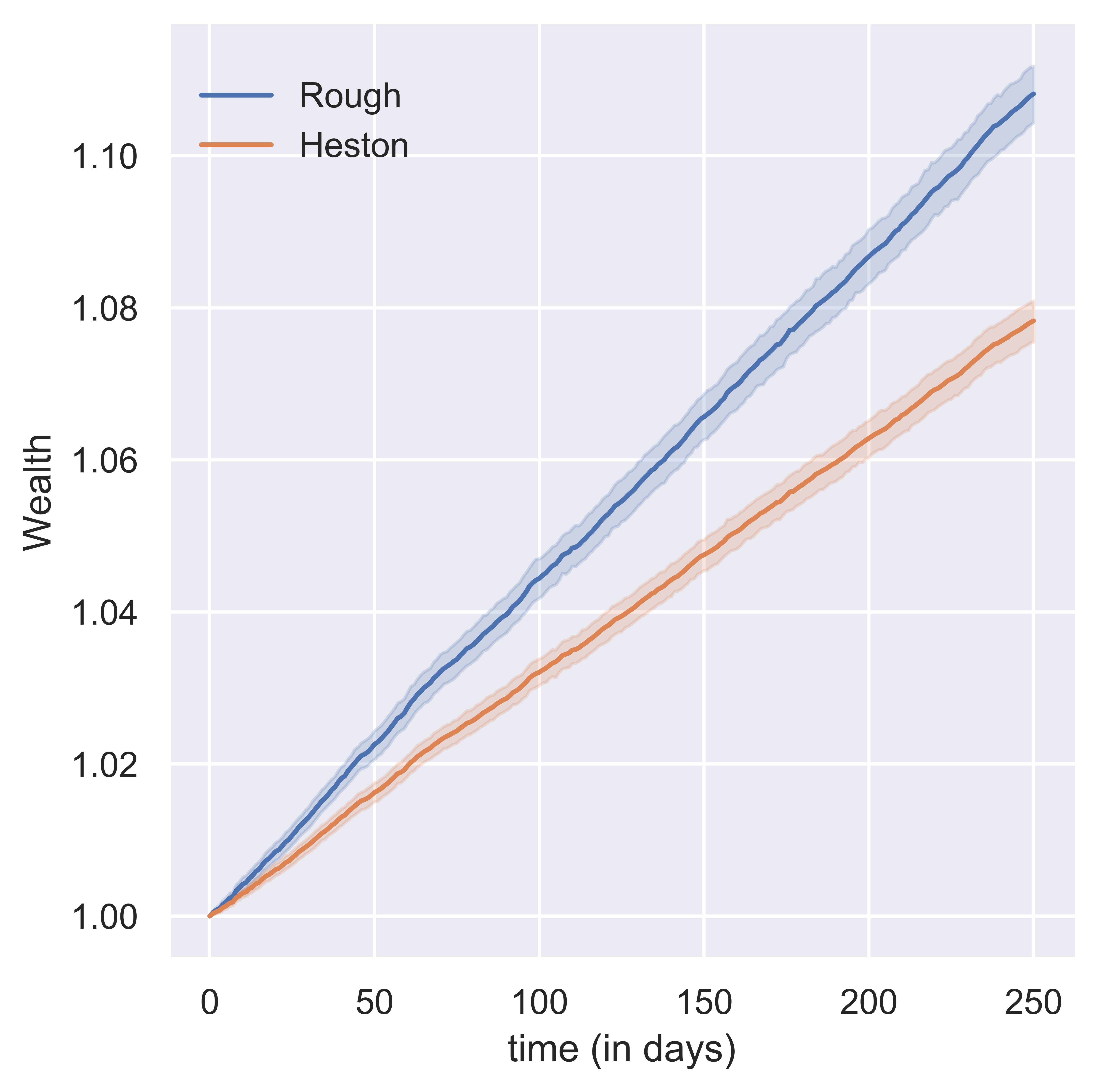}
		\subcaption{Wealth}\label{Fig:wealth}
	\end{minipage}
	\caption{Statistics for strategies and wealth. Based on 3000 simulated paths, the solid line plots the mean and the shadow area is the $95\%$ confidence interval estimated by bootstrapping. The rough Heston model suggests investing more and the terminal wealth is closer to the expected value $c=1.1163$. The parameters are the same as in Figure 3.}
\end{figure}

Recently, a trading strategy has been proposed to buy the roughest stocks and sell the smoothest stocks \cite{glasserman2019buy}. This model-free strategy aims at investments in multiple assets. Although we consider a single risky asset with a specific model, it is still interesting to compare that strategy with ours. Note that a stock is rougher for a smaller $\alpha$. Figures (\ref{Fig:u4smallsig})-(\ref{Fig:u4bigsig}) indicate that $\alpha$ is not the only factor determining the investment in a stock. The trading idea in \cite{glasserman2019buy} agrees with Figure (\ref{Fig:u4bigsig}), because the optimal investment position $u^*$ is larger for a smaller $\alpha$. However, an inconsistency occurs in Figure (\ref{Fig:u4smallsig}). Indeed, if we use the VVIX index as a proxy for the vol-of-vol, then the vol-of-vol seems larger in 2007, 2008, 2010, and 2015. The buy-rough-sell-smooth strategy \cite{glasserman2019buy} performs better in 2005, 2007, 2008, 2010, and 2014 than in other years, as shown in \cite[Figure 3]{glasserman2019buy}. This consistency suggests that vol-of-vol may also be important when roughness is considered. It would be interesting to test the performance of strategies based on roughness and vol-of-vol in a future study.

\begin{figure}[!h]
	\centering
	\begin{minipage}{0.5\textwidth}
		\centering
		\includegraphics[width=0.9\textwidth]{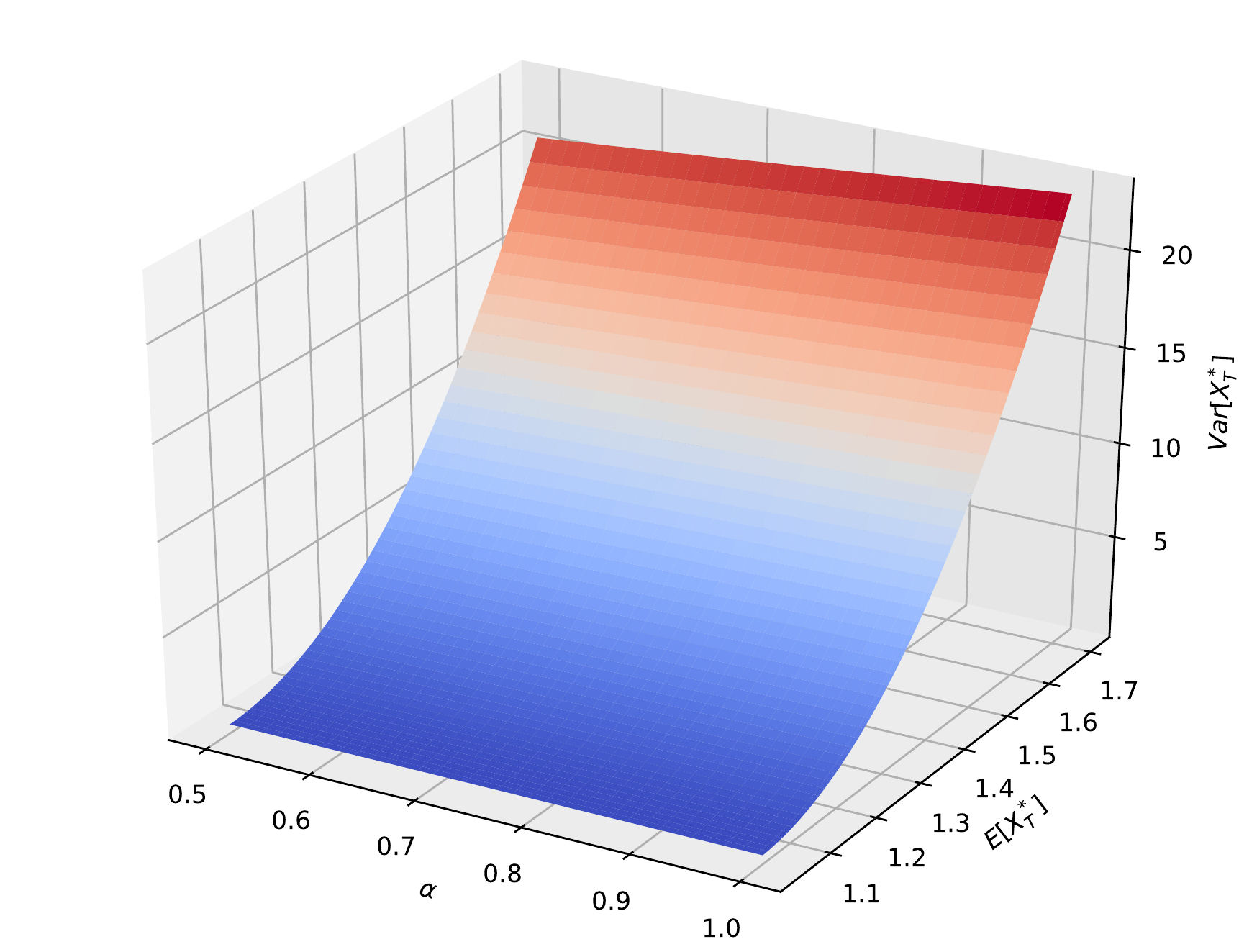}
		\subcaption{Efficient frontier}\label{Fig:effront}
	\end{minipage}%
	\begin{minipage}{0.5\textwidth}
		\centering
		\includegraphics[width=0.8\textwidth]{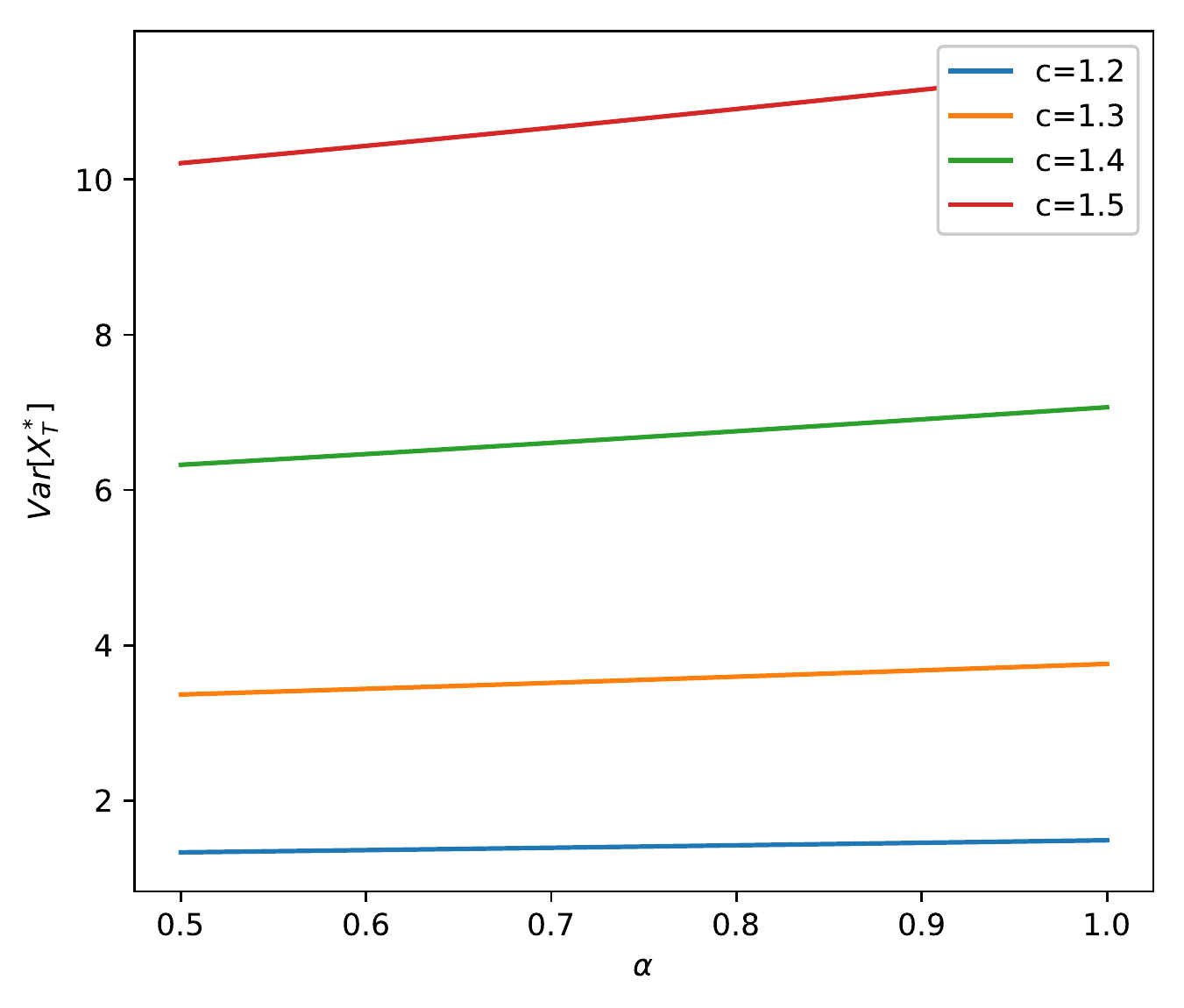}
		\subcaption{Var$[X^*_T]$ under different expected value $c$}\label{Fig:Varc}
	\end{minipage}
	\caption{Plots of the efficient frontier and variance. Roughness parameter $\alpha \in [0.5, 1]$. We set $r=0.03$, $V_0 =0.04$, $x_0 = 1$, $\phi = 0.3$, $\sigma = 0.03$, $\kappa = 0.1$, $\theta = 0.6$, $\rho = -0.7$, $T = 1$, and $c \in [x_0 e^{(r+0.01)T}, \,x_0 e^{(r+0.5)T}]$.}
\end{figure}

In Figures (\ref{Fig:effront})-(\ref{Fig:Varc}), the efficient frontier is shown for different values of $\alpha$ and expected wealth level $c$. Their relationship is clear, and the variance of the optimal wealth is reduced if $\alpha$ decreases, as $M_0$ decreases when $\alpha$ decreases and Var$[X^*_T]$ in (\ref{Eq:VarX*}) is an increasing function on $M_0$. We have also verified Assumption \ref{Assum:V} under the setting in Figures (\ref{Fig:effront})-(\ref{Fig:Varc}).

\section{Conclusion}\label{Sec:Conclusion}
To the best of our knowledge, this is the first study of the continuous-time Markowitz's mean-variance portfolio selection problem under a rough stochastic environment. We specifically focus on the Volterra Heston model. By deriving the optimal strategy and efficient frontier, we obtain further insights into the effect of roughness on them.

There are many possible future research directions. Natural considerations are the utility maximization and time-inconsistency of the MV criterion. In addition, we have already included model ambiguity with rough volatility in our research agenda.

\section*{Acknowledgements}
	The authors would like to thank two anonymous referees and the Editor for their careful reading and valuable comments, which have greatly improved the manuscript.

\appendix
\section{Solutions of Riccati-Volterra equations}\label{Appendix}
To demonstrate the existence and uniqueness of the solution to a Riccati-Volterra equation, we first rephrase the following result from a recent monograph \cite{brunner2017volterra} with more general assumptions. The underlying idea of the proof is the Picard iteration.
\begin{theorem}\label{Thm:Picard}
	Suppose kernel $K(\cdot)$ is bounded or is the fractional kernel with $\alpha \in (0, 1)$. Let $c_0, c_1, c_2$ be constant. Then there exsits $\delta>0$ such that 
	\begin{equation}
	f(t) = \int^t_0 K(t-s) \big[ c_0 + c_1 f(s) + c_2 f^2(s) \big] ds
	\end{equation}
	has a unique continuous solution $f$ on $[0, \delta]$.
\end{theorem}
\begin{proof}
	Note that quadratic function is locally Lipschitz; then according to Theorem 3.1.2 and Theorem 3.1.4 in \cite{brunner2017volterra}, the claim holds.
\end{proof}

However, $\delta$ in Theorem \ref{Thm:Picard} is not explicit. Tighter results exist if more assumptions are imposed.

We investigate $g(a,t)$ in (\ref{Eq:g}) first. Based on \cite[Theorem A.5]{gatheral2018forward}, we have
\begin{lemma}\label{Lem:g}
	Suppose {\color{black} Assumption \ref{Assum:K} holds and} $\kappa^2 - 2a\sigma^2 > 0$. Then (\ref{Eq:g}) has a unique global {\color{black} continuous} solution. Moreover, 
	\begin{equation}
	0 < g(a, t) \leq r_2(t) < w_*, \quad \forall \; t > 0,
	\end{equation}
	where $w_* \triangleq \frac{\kappa - \sqrt{\kappa^2 - 2a \sigma^2}}{\sigma^2}$ and $r_2(t) \triangleq Q^{-1}_2 \big( \int^t_0 K(s) ds \big)$; that is, the inverse function of $Q_2$, given by
	\begin{equation}
	Q_2(w) = \int^w_0 \frac{du}{a - \kappa u + \frac{\sigma^2}{2} u^2}.
	\end{equation}
\end{lemma}
\begin{proof}
	To apply the result in \cite[Theorem A.5]{gatheral2018forward}, we define
	\begin{equation*}
	H(w) = a - \kappa w + \frac{\sigma^2}{2} w^2.
	\end{equation*}
	Then $H(w)$ satisfies Assumption A.1 in \cite{gatheral2018forward} with $w_{max} \triangleq \frac{\kappa}{\sigma^2}$ and $w_*$ defined above. The claim follows from \cite[Theorem A.5 (c)]{gatheral2018forward} with $a(t) \equiv 0$ in their theorem.
\end{proof}

For the specific fractional kernel $K(t) = \frac{t^{\alpha-1}}{\Gamma(\alpha)}$,  \cite[Theorem 3.2]{eleuch2018perfect} obtains the following tighter results and the proof is based on the scaling limits of the Hawkes processes.

\begin{lemma}\label{Lem:gfractional}
	If $K(t) = \frac{t^{\alpha-1}}{\Gamma(\alpha)}$, {\color{black}$\alpha \in (1/2, 1)$}, then $g(a, t)$ in (\ref{Eq:g}) satisfies 
	\begin{equation}
	g(a, t) \leq  \frac{c}{\sigma^2} \Big[ \kappa + \frac{t^{-\alpha}}{\Gamma(1-\alpha)} + \sigma \sqrt{ a_0(t) - a} \Big],
	\end{equation}
	with $a_0(t) = \frac{1}{2\sigma^2} \big[ \kappa + \frac{t^{-\alpha}}{\Gamma(1-\alpha)} \big]^2$ and a constant $c>0$. In other words, if $ a < a_0(T)$, then Assumption \ref{Assum:V} is satisfied. 
\end{lemma}

Next, we study $\psi(\cdot)$ in (\ref{Eq:psi}). (\ref{Eq:psi}) has a unique continuous solution on some interval $[0, \delta]$ if the conditions in Theorem \ref{Thm:Picard} are satisfied. Without Theorem \ref{Thm:Picard}, we also have the following result.

\begin{lemma}\label{Lem:psi}
	{\color{black} Suppose Assumption \ref{Assum:K} holds.}
	\begin{enumerate}[label={(\arabic*).}]
		\item  If $1 - 2\rho^2 > 0$, then (\ref{Eq:psi}) has a unique global {\color{black} continuous} solution $\psi \in L^2_{loc}(\R_+, \R)$ and $\psi < 0$ for $t>0$. 
		\item  If $1 - 2\rho^2 = 0$, then (\ref{Eq:psi}) is linear and has a unique continuous solution on $[0, T]$.
		\item  If $1 - 2\rho^2 < 0$, further assume $\lambda >0$ and $\lambda^2 + 2(1-2\rho^2)\theta^2\sigma^2 >0$. Then (\ref{Eq:psi}) has a unique global {\color{black} continuous} solution. Moreover,
		\begin{equation}
		\frac{\bar w_*}{1-2\rho^2} < \frac{\bar r_2(t)}{1-2\rho^2} \leq \psi(t) < 0, \quad \forall \; t > 0,
		\end{equation}
		with $\bar w_* = \frac{\lambda - \sqrt{\lambda^2 + 2(1-2\rho^2)\theta^2\sigma^2}}{ \sigma^2}$ and $\bar r_2(t) \triangleq \bar Q^{-1}_2 \big( \int^t_0 K(s) ds \big)$, where
		\begin{equation}
		\bar Q_2(w) = \int^w_0 \frac{du}{ \frac{\sigma^2}{2} u^2 - \lambda u - (1-2\rho^2)\theta^2}.
		\end{equation}
	\end{enumerate}
\end{lemma}
\begin{proof}
	The claim in (1) follows from \cite[Theorem 7.1]{abi2017affine}. {\color{black} The continuity follows from the uniqueness of the global solution and \cite[Theorem 12.1.1]{gripenberg1990volterra}.} The claim in (2) is classic and can be found in \cite[Theorem 1.2.3]{brunner2017volterra}. For (3), we consider $ \tilde \psi = (1 - 2 \rho^2) \psi$. Then $\tilde \psi $ satisfies
	\begin{equation}\label{Eq:tildepsi}
	\tilde \psi = K * \big( \frac{\sigma^2}{2} \tilde \psi^2 - \lambda \tilde \psi - (1 - 2\rho^2) \theta^2 \big).
	\end{equation} 
	Define 
	\begin{equation}
	H(w) = \frac{\sigma^2}{2} w^2 - \lambda w - (1 - 2\rho^2) \theta^2.
	\end{equation}
	Then $\bar w_*$ is the unique root of $H(w) = 0$ on $(-\infty, \bar w_{max}]$ with $\bar w_{max} = \frac{\lambda}{\sigma^2}$. $H(w)$ satisfies Assumption A.1 in \cite{gatheral2018forward}. Therefore, \cite[Theorem A.5 (c)]{gatheral2018forward} with $a(t) \equiv 0$ implies (\ref{Eq:tildepsi}) has a unique global continuous solution and 
	\begin{equation}
	0 < \tilde \psi(t) \leq \bar r_2(t) < \bar w_*, \quad \forall \; t > 0.
	\end{equation}
	Note $ \tilde \psi = (1 - 2 \rho^2) \psi$. This gives the result desired.
\end{proof}

{\color{black}
	\section{Positivity of integrals with forward variance}\label{App:Pos}
	\begin{lemma}\label{Lem:Positive}
		Suppose Assumption \ref{Assum:K} holds. The forward variance $\xi_t(s)$ in \eqref{Eq:xi} satisfies $\int^T_t \xi_t(s) ds > 0$, $\p$-$\as$, for every $ t \in [0, T)$.
	\end{lemma}
	\begin{proof}
		As $\int^T_t \xi_t(s) ds = \tildeE [\int^T_t V_s ds| \cF_t]$ and $V_s$ is non-negative by Theorem \ref{Thm:SVSol}, it is sufficient to show that $\int^T_t V_s ds > 0$, $\p$-$\as$. 
		
		Given $t \in [0, T)$, for $\omega \in \Omega$ such that $V_s(\omega)$ is continuous in $s$, we suppose $\int^T_t V_s(\omega) ds = 0$. By the continuity of $V_s(\omega)$, $V_s(\omega) = 0$ for $s \in [t, T]$. Using the argument given in \cite[Theorem 3.5, Equation (3.8)]{abi2017affine}, for $0 < h < T - t$, we have
		\begin{align}
		V_{t + h}(\omega) =& V_0 +  \int_{0}^{t} K(t+h -s)\left(\kappa\phi - \lambda V_s (\omega) \right) d s + \int_{0}^{t} K(t+h-s) \sigma \sqrt{V_{s} (\omega)} d \tilde B_{s} (\omega) \nonumber \\
		& + \int_{t}^{t+h} K(t+h -s)\left(\kappa \phi - \lambda V_{s} (\omega) \right) d s + \int_{t}^{t+h} K(t+h-s) \sigma \sqrt{V_{s} (\omega) } d \tilde B_{s} (\omega) \nonumber \\
		\geq& \int_{t}^{t+h} K(t+h -s)\left( \kappa \phi - \lambda V_{s} (\omega) \right) d s + \int_{t}^{t+h} K(t+h-s) \sigma \sqrt{V_{s} (\omega) } d \tilde B_{s}(\omega).
		\end{align}
		As $V_s (\omega) = 0,\, s \in [t, t+h]$, then
		\begin{equation}
		V_{t + h} (\omega) \geq \kappa \phi \int_{t}^{t+h} K(t+h - s) ds > 0.
		\end{equation}
		This contradiction implies that $\int^T_t V_s ds > 0$, $\p$-$\as$, and the claim follows.
	\end{proof}
}

\end{document}